\documentclass[journal,10pt]{IEEEtran}
\makeatletter
\def\endthebibliography{%
    \def\@noitemerr{\@latex@warning{Empty `thebibliography' environment}}%
    \endlist
}
\makeatother

\usepackage{cite}

\usepackage[pdftex]{graphicx}
\usepackage[caption=false,font=footnotesize]{subfig}
\usepackage{tikz}
\usetikzlibrary{arrows.meta, positioning, shapes.geometric}
\usepackage{amsmath}
\usepackage{mathtools, cuted}
\usepackage{amssymb}
\usepackage{bm}
\usepackage{mathrsfs}
\usepackage{breqn}
\usepackage{bbm}
\usepackage{booktabs}

\usepackage{pifont}
\usepackage{xcolor}
\usepackage{url}
\usepackage{lettrine}
\usepackage{lipsum}
\usepackage{siunitx}
\usepackage{soul}
\usepackage{array}
\usepackage[inline]{enumitem}
\usepackage{epsfig}
\usepackage{filecontents}

\usepackage{algpseudocode}
\usepackage{algorithm, tabularx}
\usepackage{multirow}
\newcolumntype{L}[1]{>{\raggedright\let\newline\\\arraybackslash\hspace{0pt}}m{#1}}
\newcolumntype{C}[1]{>{\centering\let\newline\\\arraybackslash\hspace{0pt}}m{#1}}
\newcolumntype{R}[1]{>{\raggedleft\let\newline\\\arraybackslash\hspace{0pt}}m{#1}}
\newlength{\maxwidth}

\makeatletter
\newcommand{\multiline}[1]{%
	\begin{tabularx}{\dimexpr\linewidth-\ALG@thistlm}[t]{@{}X@{}}
		#1
	\end{tabularx}
}
\makeatother
\algdef{SE}[SUBALG]{Indent}{EndIndent}{}{{\algorithmicend\ }}
\algtext*{Indent}
\algtext*{EndIndent}

\usepackage{amsthm}
\newtheorem{theorem}{Theorem}

\newtheorem{lemma}[theorem]{Lemma}

\theoremstyle{remark}

\newtheorem{remark}{Remark}
\DeclareMathOperator*{\argmin}{\arg\min}

\begin{document}

\title{Toward a Receiver-Induced Channel Shaping Paradigm: FRIS-Assisted Rydberg Atomic MIMO with Quadrature-Leakage-Aware Design}

	\author{Hong-Bae Jeon,~\IEEEmembership{Member,~IEEE}, Kai-Kit Wong,~\IEEEmembership{Fellow,~IEEE}, and Chan-Byoung Chae,~\IEEEmembership{Fellow,~IEEE}
\thanks{This work was supported by Hankuk University of Foreign Studies Research Fund of 2026. \textit{(Corresponding Author: Chan-Byoung Chae)}}%
\thanks{H.-B. Jeon is with the Department of Information Communications Engineering, Hankuk University of Foreign Studies, Yong-in, 17035, Korea (e-mail: hongbae08@hufs.ac.kr).}%
		\thanks{K.-K. Wong is with the Department of Electronic and Electrical Engineering, University College London, WC1E 6BT London, U.K., and also with the Yonsei Frontier Laboratory, Yonsei University, Seoul 03722, South Korea (e-mail: kai-kit.wong@ucl.ac.uk).}
\thanks{C.-B. Chae is with the School of Integrated Technology, Yonsei University, Seoul 03722, Korea (e-mail: cbchae@yonsei.ac.kr).}
}

\maketitle
\begin{abstract}
This paper investigates a fluid reconfigurable intelligent surface (FRIS)-assisted Rydberg Atomic REceiver (RARE) architecture under magnitude-only heterodyne readout. We show that, unlike conventional coherent systems, the optimal propagation environment is fundamentally governed by the receiver’s nonlinear measurement structure. In particular, under the strong-reference regime, symbol detection is limited by residual quadrature leakage after reference alignment, motivating a receiver-induced channel shaping approach rather than conventional channel-centric optimization. Based on this insight, we formulate a signal-independent leakage minimization problem that jointly optimizes the FRIS port set, finite-resolution phase shifts, and the transmit beamformer, resulting in a nonconvex mixed discrete-continuous design. To address this, we develop an alternating-optimization (AO) framework comprising: (i) a closed-form eigenvector solution for widely-linear beamforming, (ii) cross-entropy method (CEM)-based combinatorial port selection, and (iii) coordinate-descent (CD) phase refinement with guaranteed monotonic descent. Simulation results demonstrate fast convergence and consistent bit-error-rate (BER) gains across various modulation orders and receiver dimensions. Moreover, the proposed FRIS-enabled design achieves near-exhaustive performance with significantly reduced complexity and consistently outperforms conventional RIS schemes with fixed elements, highlighting the effectiveness of spatial reconfiguration in suppressing quadrature leakage and the additional spatial degree-of-freedom (DoF) enabled by FRIS for reliable atomic-MIMO detection.
\end{abstract}


\IEEEpeerreviewmaketitle
\section{Introduction}
\label{sec:intro}
\subsection{Breaking Conventional Receiver Limits: Rydberg Atomic Receiver (RARE)}
\IEEEPARstart{T}{he} continuing evolution of wireless communications is being accelerated by unprecedented demands for extremely high data rates, ultra-reliable transmission, and low-latency connectivity, driven by the proliferation of data-intensive and latency-critical services~\cite{ss6G, T6G, NewHor, smidaFD, seman}. To meet these stringent requirements, sixth-generation (6G) wireless systems are envisioned to move beyond conventional link-level enhancements and adopt fundamentally new paradigms. These paradigms enable simultaneous improvements in spectral efficiency~\cite{yhFD}, interference mitigation in ultra-dense networks~\cite{MA}, and reliable communication under high mobility and adverse propagation conditions~\cite{hjcoop, hjaib}. At the same time, these performance targets must be achieved under strict constraints on energy consumption and hardware complexity, motivating a rethinking of both receiver architectures and propagation environment design.

Among emerging receiver technologies, Rydberg Atomic REceiver (RARE)~\cite{atomicmag, atomicjsac, rarnature, RAR_Classical_Comm_Sensing} have recently attracted significant attention as a promising candidate for next-generation wireless front-ends. Unlike conventional metallic antennas that rely on electronic circuitry for signal reception, RARE exploits the quantum-mechanical response of highly excited Rydberg atoms to incident radio-frequency (RF) electromagnetic (EM) fields. Specifically, a Rydberg atom, formed by exciting an electron to a high-lying energy level, exhibits strong dipole interactions with external RF fields. Leveraging phenomena such as electromagnetically-induced-transparency (EIT) and Autler-Townes (AT) splitting, RARE converts RF field variations into measurable optical signatures, enabling reliable recovery of transmitted signals~\cite{tqe, qprobe}. 

Several demonstrations have shown that RARE can extract diverse signal attributes with exceptional receiver sensitivity~\cite{rydamp, cfatomic, mer, rarclose2}. This capability originates from its fundamentally different reception mechanism compared to conventional RF antenna. Unlike metal conductors, where thermal agitation of charge carriers imposes an intrinsic thermal-noise limit, the atom-field interaction in RARE does not generate thermal-noise~\cite{mer, rydhard3, rarclose1}. Instead, the dominant noise source is quantum shot noise (QSN) associated with probing Rydberg states, which is several orders of magnitude lower than the classical thermal-noise floor~\cite{Harnessing_RAR_Tutorial, quanmobi}, enabling reception close to the standard quantum limit (e.g., lower than $1~\mathrm{nV cm^{-1} Hz^{-1/2}}$)~\cite{qsens, resry, qs22}. As a result, the noise-equivalent power (NEP) of RARE is approximately 20-30~dB lower than that of the most sensitive RF receiver~\cite{rqs}. Moreover, whereas conventional antennas require physical dimensions comparable to the carrier wavelength, RARE detects incident radiation through photon-atom interactions enabled by the wave-particle duality of EM-waves~\cite{foot, rarclose1, rarclose2}. As a result, its operating frequency can be flexibly tuned over ultra-wide ranges simply by selecting appropriate atomic energy-level transitions, without modifying the underlying hardware structure~\cite{efm, rarclose3, rydhard1, rydhard3}.

These unique advantages make RARE a highly promising technology for next-generation wireless reception in extremely weak electromagnetic environments. Such scenarios include satellite and space-air-ground integrated networks, where ultra-high sensitivity is critical~\cite{sagn, antsp}, thereby surpassing the inherent limitations of conventional RF receivers. Building on this capability, subsequent studies have progressively established communication and signal processing frameworks for RARE. In~\cite{atomicjsac}, the feasibility of atomic demodulation was extended to multi-user scenarios by formulating atomic single/multiple-input-multiple-output (SIMO/MIMO) detection as a biased phase-retrieval problem and developing Expectation-Maximization Gerchberg-Saxton (EM-GS) algorithm for joint symbol recovery. Furthermore,~\cite{Precoding_atomicMIMO} introduced an analytic atomic-MIMO model, revealing a fundamental departure from conventional RF MIMO model due to the nonlinear magnitude-only input-output relation inherent to atomic sensing, and proposed in-phase/quadrature (I/Q)-aware precoding strategies to approach atomic-MIMO capacity limits. In parallel, quantum-enhanced spatial sensing techniques have been developed. Specifically, a MUltiple SIgnal Classification (MUSIC)-variant angle-of-arrival (AoA) estimation algorithm, termed Quantum-MUSIC, was proposed in~\cite{qmusic}, which enables multi-user AoA estimation using RARE by recovering channel information from magnitude-only atomic measurements and exploiting subspace orthogonality for estimation. Moreover, multi-band atomic sensing was investigated in~\cite{wsat}, demonstrating that RARE can exploit multiple atomic energy-level transitions to detect multi-frequency EM signals and estimate their spatial parameters via quantum probe-laser measurements. The authors in~\cite{TCOM25_Single_RAR_AoA} develop a framework for AoA detection using a single RARE element by modeling the spatially varying susceptibility induced by interference between the incident and local-oscillator (LO) signals and deriving a closed-form relationship between probe laser transmission and AoA. Beyond algorithmic developments, fundamental receiver physics has also been investigated, wherein the authors in~\cite{saatomic} established a theoretical foundation for single-antenna atomic beamforming and showed that magnitude-only atomic measurements inherently induce a structured nonlinear detection model, fundamentally distinguishing RARE from classical coherent RF front-ends.

Complementing these theoretical advances, recent work has focused on improving the practical sensing capability and robustness of RARE. In~\cite{cfatomic}, a calibration-free receiving scheme for RARE was proposed by modeling the combined direct and alternating current (D/AC)-Stark response and retrieving field amplitude through cycle-averaged spectral fitting, achieving highly sensitive sub-MHz electric-field detection without external calibration. More recently, system-level architectures such as atomic multi-user uplink detection~\cite{RAR_MU_MIMO_Uplink} and quantum-MIMO receiver arrays~\cite{RAQ_MIMO_Multiband} have been introduced, demonstrating the feasibility of RARE for spatial multiplexing and multi-user wireless communication. Collectively, these works establish RARE as a viable and promising platform for next-generation wireless communication systems. Despite these promising advances, RARE introduces a unique structural limitation. Due to the magnitude-only heterodyne readout inherent to atomic sensing with real-valued representation, symbol detection performance is fundamentally constrained by residual quadrature leakage even after reference alignment~\cite{Precoding_atomicMIMO, RIS_atomic_MIMO}. This inherent I/Q coupling cannot be fully eliminated by conventional beamforming or precoding alone, thereby fundamentally redefining the role of environment design in atomic-MIMO systems. Unlike conventional coherent architectures, where the propagation environment is optimized primarily to strengthen the effective channel, RARE-based systems require the channel itself to be shaped in accordance with the receiver’s nonlinear measurement structure. In particular, the propagation environment must be engineered not only to alleviate the induced I/Q coupling and enable reliable symbol discrimination under magnitude-only observations of RARE, but also to actively compensate for the receiver-induced nonlinearity, thereby fundamentally enhancing detection performance.

\subsection{Reconfigurable Propagation: RIS to FRIS}
In this context, alongside advances in receiver technologies, the capability to actively reconfigure wireless propagation environments has emerged as a key enabler for performance enhancement~\cite{alexsmartris, nfris}. In particular, reconfigurable intelligent surfaces (RIS) offer programmable control over electromagnetic wave propagation through passive reflecting elements, enabling coverage improvement and interference mitigation~\cite{RIST,risspm}. However, conventional RIS architectures suffer from inherent limitations due to fixed element locations and severe double-fading attenuation in cascaded channels~\cite{DF, HBRIS, HBRIS22, vsrelay}. To overcome these constraints, fluid RIS (FRIS) have been introduced as a new paradigm that enables spatial reconfigurability through dynamic port selection~\cite{FRISlook,FRISmag}, inspired by fluid~\cite{FAS, fluidjsac, fluidtut, fluidsurv, fluidoj} and reconfigurable~\cite{trimimo, trimimo22} antenna systems. Unlike conventional RIS, FRIS employs a dense grid of candidate ports, among which only a subset is selected, thereby providing an additional spatial degree-of-freedom (DoF) beyond phase control.

Early studies on FRIS established its fundamental performance advantages over conventional RIS architectures. In particular,~\cite{FRISlook} and~\cite{FRISpa} introduced the concept of position reconfigurability and developed analytical and optimization frameworks for FRIS-assisted systems. Specifically,~\cite{FRISlook} demonstrated significant performance improvements over conventional RIS in both single- and multi-user scenarios through joint port-selection and phase optimization, while~\cite{FRISpa} provided rigorous analytical characterizations of outage probability and capacity, including tight upper-bounds under statistical channel models. Practical architectures incorporating discrete phase shifts and port selection were investigated in~\cite{FRISonoff}, demonstrating substantial performance gains with reduced hardware complexity. Further research has extended FRIS to a wide range of scenarios, demonstrating its effectiveness in mitigating fundamental propagation limitations. In particular,~\cite{FRISambc} investigated FRIS-assisted ambient backscatter communication and showed that spatial port reconfiguration can effectively alleviate the severe double-fading effect inherent in cascaded channels, resulting in substantial improvements in achievable backscatter rates compared to conventional RIS-aided systems. The benefits of FRIS in secure communications were further explored in~\cite{FRISsec} and~\cite{FRISsec22}, where dynamic port selection was leveraged to enhance secrecy performance by exploiting spatial correlation and selectively strengthening legitimate channels while suppressing eavesdropping links. More recently, extensions incorporating element-level radiation pattern reconfigurability were proposed in~\cite{FRISbp}, enabling joint optimization of port locations, radiation patterns, and beamforming. 

In addition,~\cite{FRISppd} provided a systematic investigation of the interplay between spatial position optimization and phase design in FRIS, demonstrating that spatial reconfigurability introduces an additional DoF, particularly in scenarios with aperture and correlation constraints. These results collectively demonstrate that FRIS provides an additional spatial DoF beyond conventional phase-only control, significantly improving link reliability, achievable rate, and security performance.

\subsection{Receiver-Induced Channel Shaping for Atomic MIMO}
Nevertheless, existing FRIS designs, as is also the case for conventional RIS, primarily focus on coherent receivers and do not account for receiver-specific nonlinear measurement structures. By explicitly incorporating such characteristics into spatial reconfiguration design, additional performance gains can be unlocked beyond those achievable with channel-centric optimization alone. This distinction becomes particularly critical for RARE, whose magnitude-only detection mechanism fundamentally alters the input-output relationship. 

Yet, the role of spatial environment reconfiguration in mitigating atomic quadrature leakage remains largely unexplored. Though~\cite{RIS_atomic_MIMO} introduced RIS to mitigate the phase ambiguity inherent in atomic-MIMO receiver by aligning the equivalent cascaded channel with LO phase through phase optimization by Adam, its design remains fundamentally constrained by the architectural limitations of conventional RIS. In particular, conventional RIS architectures are inherently limited in realizing such receiver-aware channel shaping. This is because it can only adjust the phase of the reflected signals with fixed element locations, which restricts its ability to control the spatial structure of the effective channel required to mitigate quadrature leakage. As a result, the achievable channel adaptation is confined to phase alignment, which is insufficient to reshape the channel in accordance with the nonlinear measurement characteristics of RARE. In contrast, FRIS introduces an additional spatial DoF through fluid port reconfiguration, enabling dynamic repositioning of the selected elements. This spatial reconfigurability allows the system to actively sculpt the effective atomic channel beyond conventional phase-only control. As a result, it provides a fundamentally more powerful mechanism to suppress quadrature leakage and achieve receiver-compatible channel structures.

These observations reveal a fundamental gap between recent advances in quantum receiver technologies and existing propagation environment design, which largely overlook the receiver-induced channel shaping requirement imposed by nonlinear atomic measurements. Motivated by this perspective, this paper proposes a quadrature-leakage-aware FRIS-assisted RARE framework that explicitly incorporates receiver-induced channel shaping. Specifically, we jointly optimize transmit beamforming, FRIS port selection, and discrete phase configuration to suppress quadrature leakage in atomic-MIMO. 

\subsection{Contributions of This Work}
The main contributions are summarized as follows:
\begin{itemize}

\item \textbf{Receiver-induced channel shaping principle for FRIS-assisted atomic-MIMO systems:}
We establish a receiver-induced channel shaping principle by explicitly accounting for the magnitude-only heterodyne readout of RARE under the strong-reference regime. Unlike conventional coherent systems that prioritize channel gain or interference suppression, we identify residual quadrature leakage as the dominant performance bottleneck and show that the effective channel must be shaped to align with the nonlinear measurement structure of RARE. Based on this insight, we derive a signal-independent quadrature-leakage metric in closed-form, providing a receiver-aware objective for joint beamforming and FRIS configuration.

\item \textbf{Joint propagation design formulation with coupled spatial and widely-linear structures:}
We then formulate a receiver-aware propagation design problem in which the FRIS configuration and transmit beamforming are jointly optimized to shape the effective atomic channel for leakage suppression. The resulting problem is fundamentally different from conventional channel-centric RIS/FRIS designs because the objective is induced by the nonlinear magnitude-only measurement of RARE, not by standard SNR or rate maximization. This leads to a highly nonconvex mixed discrete-continuous program, where combinatorial port selection, discrete phase control, and widely-linear beamforming are tightly coupled through the effective FRIS-assisted atomic channel.


\item \textbf{Structured framework with joint beamforming and FRIS configuration with provable monotonic descent:}
We develop a structured alternating-optimization (AO) framework that decomposes the problem into beamforming, port selection, and phase design. The beamformer admits a widely-linear quadratic form, yielding a closed-form eigenvector solution via real-augmented representation. For port selection, a cross-entropy method (CEM) efficiently explores the combinatorial space, while phase refinement is handled via coordinate descent (CD) with one-dimensional updates and discrete projection, ensuring monotonic leakage reduction. The overall AO procedure thus guarantees a non-increasing objective sequence.

\item \textbf{Performance validation with near-exhaustive detection gains:}
Simulation results demonstrate fast and stable convergence, along with consistent BER improvements over benchmarks. The proposed FRIS-enabled design achieves near-exhaustive detection performance with significantly reduced complexity and outperforms conventional RIS-based atomic reception, highlighting the benefit of spatial reconfigurability in suppressing quadrature leakage and enhancing atomic-MIMO detection.
\end{itemize}

\begin{figure}[t]
  \begin{center}
    \includegraphics[width=0.7\columnwidth,keepaspectratio]{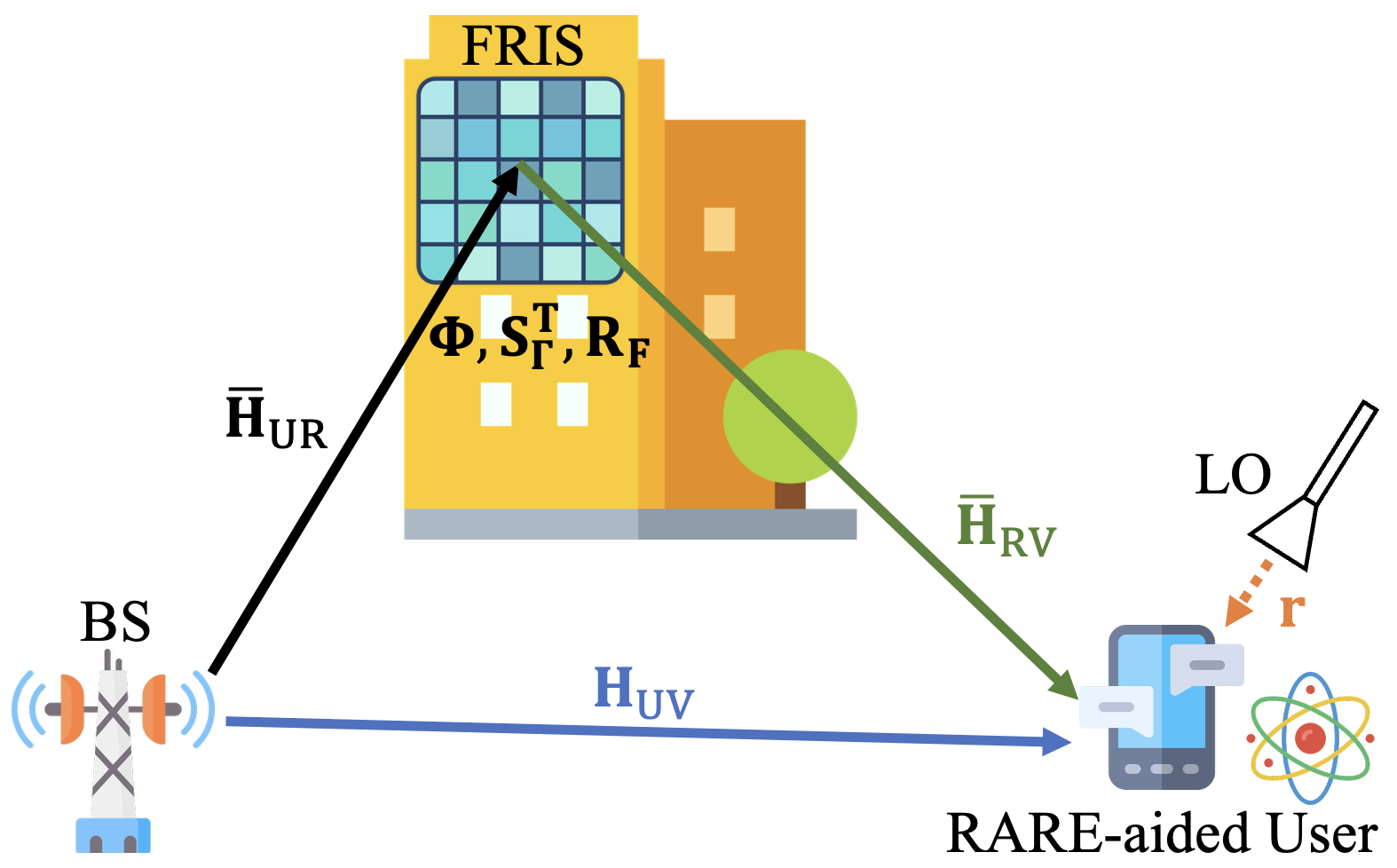}
    \caption{FRIS-assisted RARE architecture.}
    \label{fig_sys}
  \end{center}
\end{figure}
\section{System Model}
\label{sec:fris_embed}
\subsection{Comprehensive Signal Model}
As depicted in Fig.~\ref{fig_sys}, we consider an FRIS-assisted RARE architecture in which a BS equipped with $N_t$ antennas transmits signals toward an FRIS composed of $N = N_x \times N_x$ candidate ports. The reception is performed by an RARE equipped with $N_r$ vapor cells. Each FRIS element applies a discrete phase shift selected from a finite-resolution phase codebook $\mathcal{Q} \triangleq \left\{\tfrac{2\pi}{M_p}k\right\}_{k=0}^{M_p-1}$~\cite{FRISonoff}. The FRIS elements are arranged on a planar surface of size $W_x\lambda \times W_x\lambda$, where $\lambda$ denotes the carrier wavelength, resulting in an inter-element spacing of $d_x = \frac{W_x\lambda}{N_x}$. Unlike conventional RIS models that often assume independent per-element fading, FRIS candidate ports are spatially close and thus exhibit non-negligible port correlation~\cite{FRISlook, FRISnoma}. Following an isotropic scattering abstraction, we model the correlation between ports $m$ and $n$ by the zero-th order spherical Bessel function~\cite{FRISpa, FARIS}:
\begin{equation}
\label{eq:j0_corr}
\left[\mathbf R_{\rm F}\right]_{m,n}=j_0\left(\frac{2\pi}{\lambda}d_{m,n}\right),
\end{equation}
where $d_{m,n}$ is the inter-port distance between port $m$ and $n$. Accordingly, $\mathbf R_{\rm F}\in\mathbb R^{N\times N}$ is the FRIS port correlation matrix.

To consistently embed port correlation into both BS$\rightarrow$FRIS and FRIS$\rightarrow$RARE channels $\mathbf H_{\rm UR}\in\mathbb C^{N\times N_t}$ and $\mathbf H_{\rm RV}\in\mathbb C^{N_r\times N}$, respectively, we adopt the correlation-shaping model~\cite{FRISonoff}
\begin{equation}
\label{eq:HUR_corr}
\mathbf H_{\rm UR}=\mathbf R_{\rm F}^{1/2}\overline{\mathbf H}_{\rm UR},~\mathbf H_{\rm RV}=\overline{\mathbf H}_{\rm RV}\mathbf R_{\rm F}^{1/2},
\end{equation}
where $\overline{\mathbf H}_{\rm UR}\in\mathbb C^{N\times N_t}$ and $\overline{\mathbf H}_{\rm RV}\in\mathbb C^{N_r\times N}$ denote the uncorrelated channel matrices prior to correlation shaping. Herein, $\overline{\mathbf H}_{\rm UR}$ is modeled as Rician fading~\cite{FRISsec}. In contrast, $\overline{\mathbf H}_{\rm RV}$ and the direct link between BS and RARE $\mathbf H_{\mathrm{UV}}$ follow a physics-driven RARE model, and its entries will be explicitly specified in Section~\ref{subsec:rar_channel_spec}.

To model FRIS selection, we assume that only $M_o$ ports can be simultaneously selected out of the $N$ candidate ports:
\begin{equation}
\label{eq:Gamma_def}
\Gamma \subset \{1,\cdots,N\},~ |\Gamma| = M_o,
\end{equation}
where $\Gamma$ denotes the selected port set, and define the selection matrix $\mathbf S_\Gamma\in\{0,1\}^{N\times M_o}$ whose columns are canonical basis vectors indexed by $\Gamma$. The phase-shift matrix over the selected ports is
\begin{equation}
\label{eq:PhiGamma}
\mathbf \Phi\triangleq
\mathrm{diag}\left(e^{j\theta_{i_1}},\cdots,e^{j\theta_{i_{M_o}}}\right)~(\theta_{i_\ell}\in\mathcal Q),
\end{equation}
where $\Gamma=\{i_1,\cdots,i_{M_o}\}$ and $\boldsymbol\theta\triangleq[\theta_{i_1} \cdots \theta_{i_{M_o}}]^{\mathrm T}$. Accordingly, the effective correlated subchannels are
\begin{equation}
\begin{aligned}
\label{eq:HUR_Gamma_corr}
&\mathbf S_\Gamma^{\mathrm T}\mathbf H_{\rm UR}=\mathbf S_\Gamma^{\mathrm T}\mathbf R_{\rm F}^{1/2}\overline{\mathbf H}_{\rm UR}
\in\mathbb C^{M_o\times N_t}, \\
&\mathbf H_{\rm RV}\mathbf S_\Gamma
=
\overline{\mathbf H}_{\rm RV}\mathbf R_{\rm F}^{1/2}\mathbf S_\Gamma
\in\mathbb C^{N_r\times M_o}.
\end{aligned}
\end{equation}
Therefore, the end-to-end equivalent channel $\mathbf H_{\rm eq}(\Gamma,\boldsymbol\theta)$ admits the explicit correlation-aware form
\begin{equation}
\label{eq:cascaded_corr_expand}
\begin{aligned}
\mathbf H_{\rm eq}(\Gamma,\boldsymbol\theta)&\triangleq\mathbf H_{\rm RV}\mathbf S_\Gamma \mathbf \Phi \mathbf S_\Gamma^{\mathrm T}\mathbf H_{\rm UR}+\mathbf H_{\rm UV}\\
&=\overline{\mathbf H}_{\rm RV}\mathbf R_{\rm F}^{1/2}\mathbf S_\Gamma \mathbf \Phi
\mathbf S_\Gamma^{\mathrm T}\mathbf R_{\rm F}^{1/2}\overline{\mathbf H}_{\rm UR}+\mathbf H_{\rm UV},
\end{aligned}
\end{equation}
where $\mathbf H_{\rm UV}\in\mathbb C^{N_r\times N_t}$ is the BS-RARE direct channel.

\subsection{RARE-Aware Channel Specification}
\label{subsec:rar_channel_spec}
In RARE, the incident RF electric field is transduced into an optical readout through Rydberg-atom dipole interactions~\cite{atomicjsac}. Following the widely used multipath abstraction for atomic antennas, for a generic transmitter $u$ to $m$th vapor cell, the corresponding channel entry $h_{m,u}$ is modeled as~\cite{Precoding_atomicMIMO}
\begin{equation}
\label{eq:atomic_pathsum_generic}
h_{m,u}
=
\sum_{\ell=1}^{L_{m,u}}
\frac{1}{\hbar}
\boldsymbol\mu_{\rm RF}^{\mathrm T}\bm\epsilon_{m,u,\ell}
\rho_{m,u,\ell}
e^{j\varphi_{m,u,\ell}},
\end{equation}
where $\boldsymbol\mu_{\rm RF}\in\mathbb R^{3\times 1}$ is the RF electric dipole moment of Rydberg atoms and $\hbar$ is the reduced Planck constant, $L_{m,u}$ is the number of paths, $\bm\epsilon_{m,u,\ell}\in\mathbb R^{3\times 1}$ is the polarization direction,
$\rho_{m,u,\ell}$ is the path-loss term, and $\varphi_{m,u,\ell}$ is the phase shift of the $\ell$th path.

Using~\eqref{eq:atomic_pathsum_generic}, we specify $\overline{\mathbf H}_{\rm RV}\in\mathbb C^{N_r\times N}$ as
\begin{equation}
\label{eq:HRV_atomic_entry}
\begin{aligned}
&\big[\mathbf H_{\rm RV}\big]_{m,n}=\sum_{\ell=1}^{L_{m,n}}
\frac{1}{\hbar}
\boldsymbol\mu_{\rm RF}^{\mathrm T}\bm\epsilon_{m,n,\ell}
\rho_{m,n,\ell}
e^{j\varphi_{m,n,\ell}}\\
&(m=1,\cdots,N_r, n=1,\cdots,N).
\end{aligned}
\end{equation}
This expression explicitly characterizes the polarization-dependent EM coupling between the $n$th FRIS port and the $m$th vapor cell based on the underlying atomic dipole interaction. Likewise, $\mathbf H_{\rm UV}\in\mathbb C^{N_r\times N_t}$ is modeled by
\begin{equation}
\label{eq:HUV_atomic_entry}
\begin{aligned}
&\big[\mathbf H_{\rm UV}\big]_{m,t}=\sum_{\ell=1}^{L_{m,t}}\frac{1}{\hbar}\boldsymbol\mu_{\rm RF}^{\mathrm T}\bm\epsilon_{m,t,\ell}\rho_{m,t,\ell}
e^{j\varphi_{m,t,\ell}}\\
&(m=1,\cdots,N_r, t=1,\cdots,N_t),
\end{aligned}
\end{equation}
which reflects the same physics-driven atomic reception mechanism between the BS antennas and the RARE vapor cells.

\subsection{RARE Front-End Model}
\label{subsec:rar_frontend}
Following the standard heterodyne atomic reception model, the receiver observes a magnitude-only readout. Let $\mathbf x\in\mathbb C^{N_t\times 1}$ denote the transmitted complex baseband signal: $\mathbf x=\mathbf ws$ with $\mathbb E[|s|^2]=1$ and $\|\mathbf w\|_2^2= P$, and let $\mathbf r\in\mathbb C^{N_r\times 1}$ denote the received reference vector by LO~\cite{atomicjsac, Precoding_atomicMIMO}:
\begin{equation}
\label{eq:r_def}
\mathbf r \triangleq [r_1 \cdots r_{N_r}]^{\mathrm T}\in\mathbb C^{{N_r}\times 1},~r_m = \frac{s_b}{\hbar} \bm\mu_{\rm RF}^{\mathrm{T}}\bm\epsilon_{b,m} \sqrt{P_b} \rho_{b,m} e^{j\phi_{b,m}},
\end{equation}
where $s_b$ is the known reference symbol, $P_b$ is the LO power, and $\bm\epsilon_{b,m},\rho_{b,m},\phi_{b,m}$ denote the polarization direction, path-loss, and phase of the LO link to $m$th vapor cell, respectively. 

\begin{figure}[t]
	\begin{center}
		\includegraphics[width=0.9\columnwidth,keepaspectratio]%
		{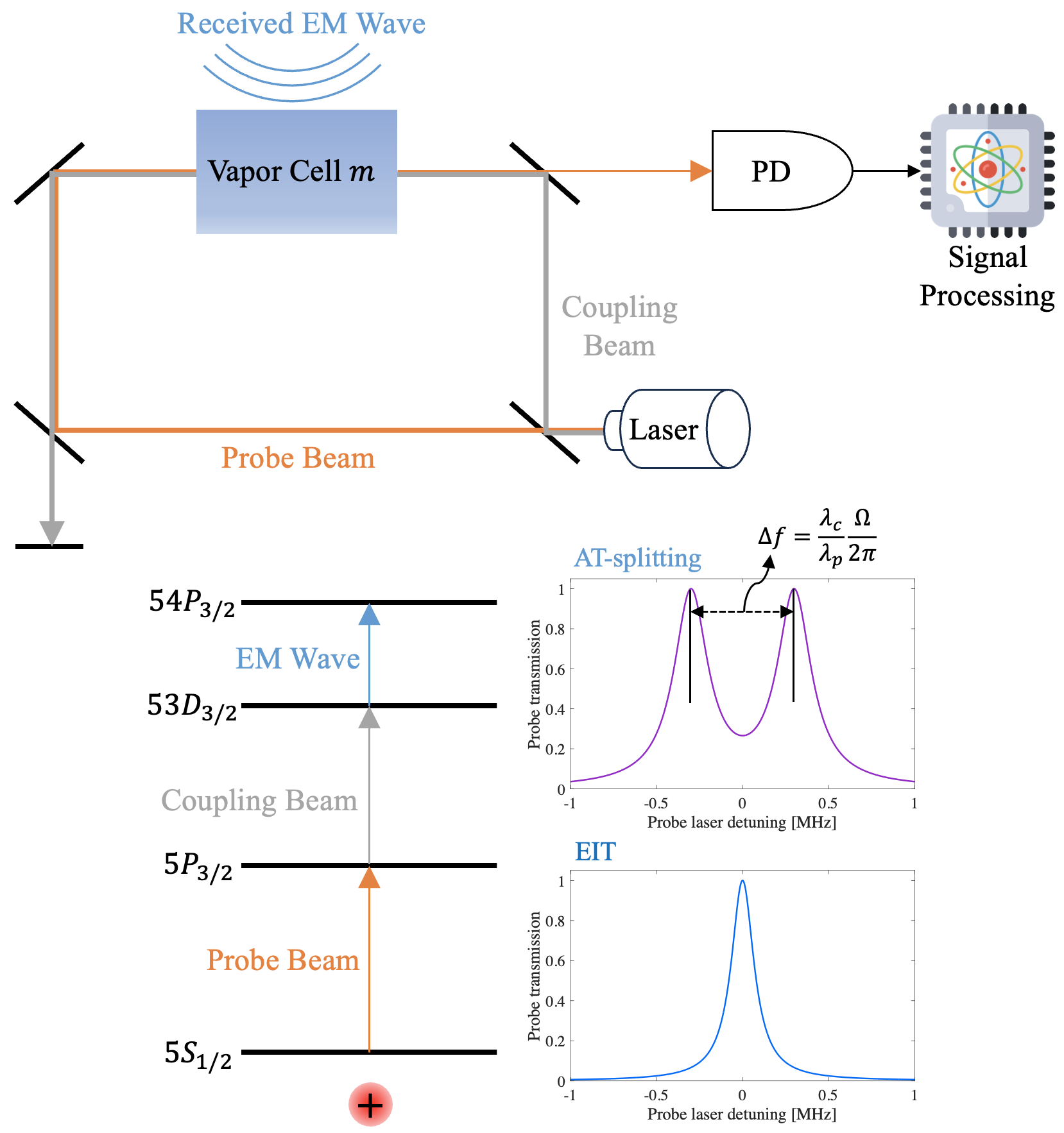}
		\caption{Illustration of signal processing in the RARE. Specifically, the incident EM field induces coupling between two highly excited Rydberg states (e.g., $53D_{3/2}$ and $54P_{3/2}$), resulting in AT splitting phenomenon. The observed spectral separation $\Delta f$ is then converted into the corresponding Rabi frequency $\Omega$ according to~\eqref{rabidef}.}
		\label{fig_rabi}
	\end{center}
\end{figure}
As illustrated in Fig.~\ref{fig_rabi}, the incident RF electric field interacts with the Rydberg atoms and induces EIT and AT splitting~\cite{efm}. The resulting AT-splitting interval $\{\Delta f_m\}$ is directly proportional to the corresponding Rabi frequency $\{\Omega_m\}$~\cite{Precoding_atomicMIMO, fox}, given by
\begin{equation}
\label{rabidef}
\Delta f_m=\frac{\lambda_c}{\lambda_p}\frac{\Omega_m}{2\pi}~(m=1, \cdots, N_r),
\end{equation}
where $\lambda_c$ and $\lambda_p$ denote the wavelengths of the coupling and probe lasers, respectively. Consequently, by measuring $\Delta f_m$, the corresponding $\Omega_m$ can be recovered using~\eqref{rabidef}. Furthermore, since $\Omega_m$ equals to the magnitude of the received electric field, this enables indirect retrieval of $s$ through atomic sensing~\cite{qmusic, atomicjsac}. Accordingly, the atomic receiver measurement can be modeled as
\begin{equation}
\label{eq:rar_magnitude_model}
\mathbf y=\left|\mathbf H_{\rm eq}(\Gamma,\boldsymbol\theta)\mathbf x + \mathbf r + \mathbf n\right|,
\end{equation}
where $\mathbf n\sim\mathcal{CN}(\mathbf 0,\sigma^2\mathbf I_{N_r})$ denotes the QSN, and the absolute value is applied element-wise. Each real-valued entry of $\mathbf y$ corresponds to $\{\Omega_m\}$.

Herein, since the LO is typically placed much closer to the vapor cells than the remote transmitter, each entry of $\mathbf r$ can dominate the wireless term and noise. In the strong-reference regime, $|r_m|\gg |[\mathbf H_{\rm eq}\mathbf x]_m+n_m|$, we can linearize~\eqref{eq:rar_magnitude_model} as~\cite{Precoding_atomicMIMO}
\begin{equation}
\label{eq:rar_strong_ref_linear}
\mathbf y-|\mathbf r|
\approx
\Re\left\{\mathbf{H}_{\rm eq}\mathbf x\circ e^{-j\angle \mathbf r}\right\}
+
\widetilde{\mathbf n},
\end{equation}
where $\widetilde{\mathbf n}\sim\mathcal N\left(\mathbf 0,\tfrac{\sigma^2}{2}\mathbf I_{N_r}\right)$.
\begin{remark}
\label{r1}
The linearization in~\eqref{eq:rar_strong_ref_linear} highlights a key distinction from conventional coherent MIMO systems. In standard RF receivers, the propagation environment is typically optimized to enhance channel gain, improve conditioning, or suppress interference~\cite{RISEE, HBRIS}. In contrast, for RARE with magnitude-only heterodyne readout in~\eqref{eq:rar_magnitude_model}, the effective channel quality depends not only on its strength but also on its alignment with the receiver’s nonlinear measurement structure after reference alignment. Consequently, the design objective is no longer purely channel-centric; instead, the environment must be shaped so that the resulting atomic channel is inherently compatible with real-part detection. This leads to a receiver-induced channel shaping perspective, where the residual quadrature component becomes the primary structural distortion to be minimized.
\end{remark}
\subsection{Quadrature-Leakage-Aware Environment Design}
\label{subsec:quadrature_formulation}
Under the strong-reference regime in~\eqref{eq:rar_strong_ref_linear}, the measurement is governed by the in-phase component $\Re\left\{\mathbf{H}_{\rm eq}\mathbf x\circ e^{-j\angle \mathbf r}\right\}$, from which it directly follows that the corresponding quadrature component is given by\begin{equation}
\label{eq:quad_vec_def}
\begin{aligned}
\mathbf z(\Gamma,\boldsymbol\theta, \mathbf w)
&\triangleq
\Im\left\{
\mathbf H_{\rm eq}(\Gamma,\boldsymbol\theta)\mathbf x
\circ e^{-j\angle\mathbf r}
\right\}\\
&=\Im\left\{\mathbf D_r \mathbf H_{\rm eq}(\Gamma,\boldsymbol\theta)\mathbf x\right\},
\end{aligned}
\end{equation}
where $\mathbf D_r \triangleq \mathrm{diag}\big(e^{-j\angle\mathbf r}\big)$, where $e^{-j(\cdot)}$ is taken element-wise, is the diagonal reference-alignment matrix. Thereby, if it vanishes, then $\mathbf a(\Gamma,\boldsymbol\theta,\mathbf w)\triangleq \mathbf D_r \mathbf H_{\rm eq}(\Gamma,\boldsymbol\theta)\mathbf x$ is purely real-valued and thus
\begin{equation}
\label{prpr}
\mathbf y-|\mathbf r|
\approx
\mathbf a(\Gamma,\boldsymbol\theta,\mathbf w)
+\widetilde{\mathbf n}
=
\mathbf D_r \mathbf H_{\rm eq}(\Gamma,\boldsymbol\theta)\mathbf x+\widetilde{\mathbf n}.
\end{equation}
Therefore, up to the known $\mathbf D_r$, the nonlinear magnitude-only atomic observation reduces to the same linear coherent form as a known-phase receiver. This shows that the zero-quadrature condition is an ideal structural target for atomic-MIMO detection. Since exact fulfillment of this condition is generally impossible, a natural design strategy is to minimize the residual quadrature component, thereby driving the atomic readout as close as possible to the coherent known-phase model. This will be compared in Section~\ref{beco}.

Since $\mathbf x$ varies across channel uses, we adopt a signal-independent design criterion based on the average quadrature-leakage energy:
\begin{equation}
\label{eq:L_def}
\mathcal L(\Gamma,\boldsymbol\theta, \mathbf w)
\triangleq
\mathbb E_{\mathbf x}
\Big[
\|
\Im\{\mathbf D_r \mathbf H_{\rm eq}(\Gamma,\boldsymbol\theta)\mathbf x\}
\|_2^2
\Big].
\end{equation}
Specifically, $\mathcal L$ serves as a receiver-aware measure of structural compatibility. It quantifies the extent to which the effective FRIS-shaped atomic channel deviates from the ideal real-valued structure imposed by the strong-reference regime.

To evaluate~\eqref{eq:L_def} in closed-form, we define the covariance and pseudo-covariance matrices of $\mathbf x=\mathbf ws$, respectively:
\begin{equation}
\label{eq:cov_pseudocov}
\mathbf Q \triangleq \mathbb E[\mathbf x \mathbf x^{*}]=\mathbf w \mathbf w^{*},~\mathbf P \triangleq \mathbb E[\mathbf x \mathbf x^{\mathrm{T}}]=\mathbf w \mathbf w^{\mathrm{T}}\kappa,
\end{equation}
with $\kappa\triangleq\mathbb E[s^2]$. Let
\begin{equation}
\label{eq:A_def}
\mathbf A(\Gamma,\boldsymbol\theta)
\triangleq
\mathbf D_r \mathbf H_{\rm eq}(\Gamma,\boldsymbol\theta),
~
\mathbf y_a \triangleq \mathbf A \mathbf x.
\end{equation}
Using the identity $\Im\{\mathbf y_a\}=\frac{1}{2j}(\mathbf y_a - \mathbf y_a^*)$, the squared norm of $\Im\{\mathbf y_a\}$ can be expressed as
\begin{equation}
\label{eq:imag_norm_expand}
\|\Im\{\mathbf y_a\}\|_2^2=-\frac{1}{4}
(\mathbf y_a-\mathbf y_a^*)^{\mathrm{T}}
(\mathbf y_a-\mathbf y_a^*).
\end{equation}
Taking expectation and expanding each term yields
\begin{equation}
\label{eq:imag_expect_expand}
\mathbb E
\big[
\|
\Im\{\mathbf y_a\}
\|_2^2
\big]
=
\frac{1}{2}
\mathrm{tr}\big(
\mathbb E[\mathbf y_a\mathbf y_a^{*}]
\big)
-
\frac{1}{2}
\Re
\left\{
\mathrm{tr}\big(
\mathbb E[\mathbf y_a\mathbf y_a^{\mathrm{T}}]
\big)
\right\}.
\end{equation}
Substituting $\mathbf y_a=\mathbf A\mathbf x$ and using~\eqref{eq:cov_pseudocov} gives
\begin{equation}
\label{eq:cov_y}
\mathbb E[\mathbf y_a\mathbf y_a^{*}]
=
\mathbf A \mathbf Q \mathbf A^{*},
~
\mathbb E[\mathbf y_a\mathbf y_a^{\mathrm{T}}]
=
\mathbf A \mathbf P \mathbf A^{\mathrm{T}}.
\end{equation}
Therefore,
\begin{equation}
\label{eq:L_trace_A}
\mathcal L(\Gamma,\boldsymbol\theta, \mathbf w)
=
\frac{1}{2}
\left(
\mathrm{tr}(\mathbf A \mathbf Q \mathbf A^{*})
-
\Re\{\mathrm{tr}(\mathbf A \mathbf P \mathbf A^{\mathrm{T}})\}
\right).
\end{equation}
Finally, substituting $\mathbf A=\mathbf D_r \mathbf H_{\rm eq}$ and using the cyclic invariance of the trace and the diagonal structure of $\mathbf D_r$ yields
\begin{equation}
\label{eq:L_final}
\mathcal L(\Gamma,\boldsymbol\theta, \mathbf w)
=
\frac{1}{2}
\left(
\mathrm{tr}
\big(
\mathbf H_{\rm eq} \mathbf Q \mathbf H_{\rm eq}^{*}
\big)
-
\Re
\left\{
\mathrm{tr}
\big(
\mathbf H_{\rm eq}
\mathbf P
\mathbf H_{\rm eq}^{\mathrm{T}}
\mathbf D_r^2
\big)
\right\}
\right).
\end{equation}
Accordingly, the corresponding optimization problem is formulated as
\begin{equation}
\label{prob:fris_quad_leakage}
\begin{aligned}
\min_{\Gamma,\boldsymbol\theta, \mathbf w}~ \eqref{eq:L_final}~
\text{s.t.}~ |\Gamma|=M_o,~\theta_{i_\ell}\in\mathcal Q~(\forall i_\ell\in\Gamma),~\|\mathbf w\|_2^2 =P.
\end{aligned}
\end{equation}
Problem~\eqref{prob:fris_quad_leakage} is a highly nonconvex mixed discrete–continuous optimization problem, where combinatorial port selection, finite-resolution phase constraints, and a continuous beamformer are tightly coupled through the equivalent channel and pseudo-covariance structure. As a result, exhaustive search becomes computationally intractable. This intricate coupling reflects the fact that spatial reconfiguration via FRIS directly governs the propagation-induced phase profile through joint port selection and phase control. As a result, quadrature leakage is not merely a signal processing artifact, but a propagation-dependent phenomenon that links the receiver’s nonlinear measurement structure with environment design.

Motivated by this observation, we develop an efficient AO framework that decomposes the original problem into tractable subproblems with respect to $\mathbf w$, $\Gamma$, and $\boldsymbol\theta$, each admitting either a closed-form solution or an efficient iterative update.

\section{Proposed AO Framework}
\label{subsec:ao_joint}
\subsection{Update of $\mathbf w$ for fixed $(\Gamma,\boldsymbol\theta)$}
Let $\mathbf H\triangleq \mathbf H_{\rm eq}(\Gamma,\boldsymbol\theta)$ and $\mathbf D\triangleq \mathbf D_r^2$. Then~\eqref{prob:fris_quad_leakage} reduces (up to a constant factor) to the widely-linear quadratic form
\begin{equation}
\min_{\|\mathbf w\|_2^2= P}
~
\mathbf w^{*}\mathbf R\mathbf w
-
\Re\{\kappa\mathbf w^{\mathrm{T}}\mathbf B\mathbf w\},
\label{prob:w_update}
\end{equation}
where $\mathbf R\triangleq \mathbf H^{*}\mathbf H$ and $\mathbf B\triangleq \mathbf H^{\mathrm{T}}\mathbf D\mathbf H$.

To obtain a closed-form update, we employ the real-augmented representation
$\mathbf w=\mathbf w_{\rm R}+j\mathbf w_{\rm I}$ and $\tilde{\mathbf w}\triangleq
[\mathbf w_{\rm R}^{\mathrm{T}}~\mathbf w_{\rm I}^{\mathrm{T}}]^{\mathrm{T}}
\in\mathbb R^{2N_t}$.
Define the real and imaginary parts
\begin{equation}
\label{rij}
\mathbf R=\mathbf R_{\rm R}+j\mathbf R_{\rm I},~\kappa\mathbf B=\mathbf C_{\rm R}+j\mathbf C_{\rm I},
\end{equation}
where $\mathbf R_{\rm R},\mathbf R_{\rm I},\mathbf C_{\rm R},\mathbf C_{\rm I}\in\mathbb R^{N_t\times N_t}$. Using standard real-valued representations of complex quadratic forms, the objective in~\eqref{prob:w_update} can be equivalently written as
\begin{equation}
\label{csw}
\mathbf w^{*}\mathbf R\mathbf w-\Re\{\kappa\mathbf w^{\mathrm{T}}\mathbf B\mathbf w\}=\tilde{\mathbf w}^{\mathrm{T}} \mathbf G
\tilde{\mathbf w},
\end{equation}
where the symmetric real matrix $\mathbf G\in\mathbb R^{2N_t\times 2N_t}$ is explicitly given by
\begin{equation}
\label{eq:G_explicit}
\mathbf G
=
\begin{bmatrix}
\mathbf R_{\rm R}-\mathbf C_{\rm R}
&
-\mathbf R_{\rm I}+\mathbf C_{\rm I}
\\
\mathbf R_{\rm I}+\mathbf C_{\rm I}
&
\mathbf R_{\rm R}+\mathbf C_{\rm R}
\end{bmatrix}.
\end{equation}
Therefore, the beamformer update reduces to the real-valued quadratic program
\begin{equation}
\min_{\|\tilde{\mathbf w}\|_2^2=P}
~
\tilde{\mathbf w}^{\mathrm{T}}\mathbf G\tilde{\mathbf w},
\label{prob:w_update_real}
\end{equation}
whose global minimizer is given by
\begin{equation}
\label{gmw}
\tilde{\mathbf w}^\star=\sqrt{P}\mathbf v_{\min}(\mathbf G),
\end{equation}
where $\mathbf v_{\min}(\mathbf G)$ denotes the unit-norm eigenvector corresponding to the smallest eigenvalue of $\mathbf G$.
The complex $\mathbf w^\star$ is then recovered by recombining the real and imaginary parts of $\tilde{\mathbf w}^\star$.

\subsection{Update of $\Gamma$ for Fixed $(\mathbf w,\boldsymbol\theta)$}
\label{subsubsec:cem_update_Gamma}
For a fixed $\boldsymbol\theta$, denote the objective value as $f(\Gamma)\triangleq \mathcal L(\Gamma,\boldsymbol\theta, \mathbf w)$. To explore the combinatorial search space efficiently, we employ the CEM, which adaptively updates a probabilistic sampling distribution to concentrate on promising regions of the solution space~\cite{FRISonoff, FRISsec, CEM}. We first parameterize a sampling distribution over ports by a probability mass function (PMF):
\begin{equation}
\label{eq:pmf_def}
\mathbf p^{(k)}=[p_1^{(k)}~\cdots~p_N^{(k)}]^{\mathrm{T}},~ p_n^{(k)}\in(0,1),~\sum_{n=1}^N p_n^{(k)}=1,
\end{equation}
where $k$ denotes the CEM iteration index. At iteration $k$, we draw $K$ independent and identically distributed (i.i.d.) candidate sets $\{\Gamma_i^{(k)}\}_{i=1}^K$ by sampling $M_o$ distinct indices {without replacement} according to $\mathbf p^{(k-1)}$. 
For each realization $\Gamma_i^{(k)}$, we evaluate $f(\Gamma_i^{(k)})$ by substituting 
$\mathbf S_{\Gamma_i^{(k)}}$ and $\boldsymbol\theta_{\Gamma_i^{(k)}}$ into $\mathbf H_{\rm eq}(\Gamma,\boldsymbol\theta)$ and then into~\eqref{eq:L_final}.
We sort $\{f(\Gamma_i^{(k)})\}_{i=1}^K$ in ascending order and define the elite index set
\begin{equation}
\label{eq:elite_set}
\mathcal E^{(k)}\triangleq \left\{i :  f(\Gamma_i^{(k)})\le f_{(K_e)}^{(k)}\right\},
~ 
K_e\triangleq \lceil \rho K\rceil,
\end{equation}
where $f_{(K_e)}^{(k)}$ is the $K_e$th order statistic and $\rho\in(0,1)$ is the elite ratio.
The PMF is then updated by maximizing the likelihood of the elite samples, yielding the closed-form update
\begin{equation}
\label{eq:pmf_update_raw}
\widehat p_n^{(k)}
=
\frac{1}{K_e}\sum_{i\in\mathcal E^{(k)}}\mathbbm 1\{n\in\Gamma_i^{(k)}\}~(n=1,\cdots,N),
\end{equation}
where $\mathbbm 1(\cdot)$ is the indicator function, followed by smoothing:
\begin{equation}
\label{eq:pmf_smooth}
p_n^{(k)} = (1-\alpha)p_n^{(k-1)}+\alpha\widehat p_n^{(k)}~ (\alpha\in(0,1]),
\end{equation}
to avoid premature convergence.
After $K_{\rm CEM}$ iterations, we output the port set by selecting the $M_o$ indices with the largest probabilities:
\begin{equation}
\label{eq:Gamma_from_p}
\Gamma^{+}=\mathrm{Top}_{M_o}\big(\mathbf p^{(K_{\rm CEM})}\big).
\end{equation}
\subsection{Update of $\boldsymbol\theta$ for fixed $(\Gamma, \mathbf w)$}
\label{subsubsec:theta_update_rigorous}
For a fixed $\Gamma=\{i_1,\cdots,i_{M_o}\}$, we optimize $\boldsymbol\theta=[\theta_{i_1} \cdots \theta_{i_{M_o}}]^{\mathrm{T}}$ by exploiting the single-coordinate structure of $\mathbf H_{\rm eq}(\Gamma,\boldsymbol\theta)$. Let $\bm\phi\triangleq [e^{j\theta_{i_1}} \cdots e^{j\theta_{i_{M_o}}}]^{\mathrm{T}}\in\mathbb C^{M_o}$ denote the unit-modulus phase vector restricted to the selected ports, with each entry constrained to the discrete alphabet $\mathcal A\triangleq\{e^{j q} : q\in\mathcal Q\}$. Define the FRIS-induced term
\begin{equation}
\label{eq:Heq_affine_phi}
\mathbf H_{\rm eq}(\Gamma,\boldsymbol\theta)
=
\mathbf H_{\rm UV}
+
\sum_{\ell=1}^{M_o} \phi_\ell \mathbf G_\ell,
\end{equation}
where $\phi_\ell=e^{j\theta_{i_\ell}}$ and
\begin{equation}
\label{eq:Gell_def}
\mathbf G_\ell
\triangleq
\mathbf H_{\rm RV}\mathbf e_{i_\ell}\mathbf e_{i_\ell}^{\mathrm{T}}\mathbf H_{\rm UR}
 \in\mathbb C^{N_r\times N_t},
\end{equation}
with $\mathbf e_{i_\ell}$ denoting the $i_\ell$th canonical basis vector.
The decomposition~\eqref{eq:Heq_affine_phi} follows directly from
$\mathbf H_{\rm RV}\mathbf S_\Gamma \mathbf\Phi \mathbf S_\Gamma^{\mathrm{T}}\mathbf H_{\rm UR}
=\sum_{\ell=1}^{M_o} e^{j\theta_{i_\ell}}\mathbf H_{\rm RV}\mathbf e_{i_\ell}\mathbf e_{i_\ell}^{\mathrm{T}}\mathbf H_{\rm UR}$.

For notational simplicity, define $\mathbf Y(\bm\phi)\triangleq \mathbf H_{\rm eq}(\Gamma,\boldsymbol\theta)$. Then~\eqref{eq:L_final} can be rewritten as
\begin{equation}
\label{eq:L_phi_form}
\mathcal L(\bm\phi)=\frac12\Big(\mathrm{tr}(\mathbf Y(\bm\phi)\mathbf Q\mathbf Y(\bm\phi)^{*})-\Re\{\mathrm{tr}(\mathbf Y(\bm\phi)\mathbf P\mathbf Y(\bm\phi)^{\mathrm{T}}\mathbf D)\}
\Big).
\end{equation}
Now fix all entries of $\bm\phi$ except $\phi_\ell$ and let
\begin{equation}
\label{eq:Y_minus_ell}
\mathbf Y_{-\ell}\triangleq \mathbf H_{\rm UV}+\sum_{m\neq \ell}\phi_m\mathbf G_m,
\end{equation}
so that $\mathbf Y(\bm\phi)=\mathbf Y_{-\ell}+\phi_\ell\mathbf G_\ell$.
Substituting $\mathbf Y_{-\ell}+\phi_\ell\mathbf G_\ell$ into~\eqref{eq:L_phi_form} and expanding the trace terms, one can collect all $\phi_\ell$-dependent terms into the following scalar form:
\begin{equation}
\label{eq:L_ell_scalar_form}
\mathcal L(\bm\phi)=C+\Re\{\alpha_\ell\phi_\ell\}+\Re\{\beta_\ell\phi_\ell^2\}~(|\phi_\ell|=1),
\end{equation}
where $C$ is independent of $\phi_\ell$, and the coefficients $\alpha_\ell$ and $\beta_\ell$ are given by
\begin{equation}
\begin{aligned}
\label{eq:alpha_beta_def}
\alpha_\ell=&\mathrm{tr}\left(\mathbf G_\ell \mathbf Q \mathbf Y_{-\ell}^{*}\right)-\frac{1}{2}\kappa\mathrm{tr}\left(\mathbf G_\ell\mathbf w\mathbf w^{\mathrm{T}}\mathbf Y_{-\ell}^{\mathrm{T}}\mathbf D\right)\\
&-\frac{1}{2}\kappa\mathrm{tr}\left(\mathbf Y_{-\ell}\mathbf w\mathbf w^{\mathrm{T}}\mathbf G_\ell^{\mathrm{T}}\mathbf D\right),\\
\beta_\ell=&-\frac{1}{2}\kappa\mathrm{tr}\left(\mathbf G_\ell\mathbf w\mathbf w^{\mathrm{T}}\mathbf G_\ell^{\mathrm{T}}\mathbf D\right),
\end{aligned}
\end{equation}
The form~\eqref{eq:L_ell_scalar_form} follows because the first trace term in~\eqref{eq:L_phi_form} is quadratic in $\mathbf Y(\bm\phi)$ and thus contributes at most linear terms in $\phi_\ell$ under $|\phi_\ell|=1$, while the second (pseudo-covariance) trace term involves $\mathbf Y(\bm\phi)^{\mathrm{T}}$ and yields both linear and quadratic terms in $\phi_\ell$. Hence under~\eqref{eq:L_ell_scalar_form}, the $\ell$th phase update reduces to the scalar problem
\begin{equation}
\label{eq:phi_cont_problem}
\min_{|\phi|=1}~
g(\phi)\triangleq \Re\{\alpha_\ell \phi\}+\Re\{\beta_\ell \phi^2\}.
\end{equation}
Let $\phi\triangleq e^{j\theta}$ with $\theta\in[0, 2\pi]$. Then
\begin{equation}
\label{eq:g_u}
\psi(\theta)\triangleq g(e^{j\theta})=\Re\{\alpha_\ell e^{j\theta}+\beta_\ell e^{j2\theta}\}.
\end{equation}
A stationary point satisfies $\frac{d}{d\theta}g(e^{j\theta})=0$. Using
\begin{equation}
\label{ddt}
\frac{d}{d\theta}\Re\{ce^{jm\theta}\}=-m\Im\{ce^{jm\theta}\},
\end{equation}
we obtain the stationarity condition with $e^{j\theta}=u$:
\begin{equation}
\label{eq:stationary_im}
\Im\{\alpha_\ell u+2\beta_\ell u^2\}=0.
\end{equation}
Condition~\eqref{eq:stationary_im} is equivalent to requiring $\alpha_\ell u+2\beta_\ell u^2$ to be real-valued, i.e.,
\begin{equation}
\label{eq:real_condition}
\alpha_\ell u+2\beta_\ell u^2 = \left(\alpha_\ell u+2\beta_\ell u^2\right)^*= \alpha_\ell^* u^{-1}+2\beta_\ell^* u^{-2}.
\end{equation}
Multiplying both sides of~\eqref{eq:real_condition} by $u^2$ yields the quartic equation
\begin{equation}
\label{eq:quartic_u}
2\beta_\ell u^4+\alpha_\ell u^3-\alpha_\ell^* u-2\beta_\ell^*=0.
\end{equation}
Let $\mathcal U_\ell$ denote the set of roots of~\eqref{eq:quartic_u} that lie on the unit circle:
\begin{equation}
\label{eq:U_set}
\mathcal U_\ell \triangleq \{u : u \text{ solves~\eqref{eq:quartic_u} and } |u|=1\}.
\end{equation}
Then the global minimizer of~\eqref{eq:phi_cont_problem} is obtained by evaluating $g(u)$ over the finite candidate set $\mathcal U_\ell$:
\begin{equation}
\label{eq:u_star}
\phi_\ell^{\rm cont} \in \argmin_{u\in \mathcal U_\ell} \Re\{\alpha_\ell u+\beta_\ell u^2\}.
\end{equation}
Finally, the discrete phase is obtained by projecting onto $\mathcal A$:
\begin{equation}
\label{eq:phi_quantize}
\phi_\ell^+=\Pi_{\mathcal A}\left(\phi_\ell^{\rm cont}\right)=e^{j\frac{2\pi}{M_p}\mathrm{round}\left(\frac{M_p}{2\pi}\angle \phi_\ell^{\rm cont}\right)},
\end{equation}
and we set $\theta_{i_\ell}\leftarrow \angle \phi_\ell^+$.
\begin{remark}
\label{r11}
Since $\psi(\theta)$ is continuous on the compact interval $[0,2\pi]$, it attains a global minimum at some $\theta^\star\in[0,2\pi]$. Moreover, $\psi$ is continuously differentiable and $2\pi$-periodic with $\psi(0)=\psi(2\pi)$, so $\theta^\star$ can be taken as an interior minimizer (equivalently on the circle), and thus it satisfies the first-order necessary condition $\psi'(\theta^\star)=0$. Letting $u^\star=e^{j\theta^\star}$ yields~\eqref{eq:stationary_im}, which is equivalent to~\eqref{eq:quartic_u}. Therefore,~\eqref{eq:quartic_u} admits at least one unit-modulus root, i.e., $\mathcal U_\ell\neq\emptyset$. This theoretical existence result is further corroborated numerically in Section~\ref{rpf}, where we explicitly visualize the feasible quartic roots on the complex plane.
\end{remark}

Having obtained the closed-form $\ell$th update rule, we now summarize the overall CD update framework for refining $\boldsymbol\theta$ (equivalently, $\bm\phi=[\phi_1 \cdots\phi_{M_o}]^{\mathrm{T}}$) under a fixed $\Gamma$. We first initialize $\bm\phi^{(0)}$. Then we denote the phase-refinement iteration index by $t$. At $\tau$th iteration, for each coordinate $\ell\in\{1,\cdots,M_o\}$, define
\begin{equation}
\label{eq:Y_minus_ell_recap}
\mathbf Y_{-\ell}^{(t)}
\triangleq
\mathbf H_{\rm UV}+\sum_{m\neq \ell}\phi_m^{(\tau)}\mathbf G_m,
~
\mathbf Y^{(\tau)}(\bm\phi)=\mathbf Y_{-\ell}^{(\tau)}+\phi_\ell \mathbf G_\ell,
\end{equation}
and compute $(\alpha_\ell^{(\tau)},\beta_\ell^{(\tau)})$ associated with the $\ell$th coordinate as in~\eqref{eq:L_ell_scalar_form}-\eqref{eq:alpha_beta_def}. Then, the continuous minimizer $\phi_{\ell,\rm cont}^{(\tau+)}$ is obtained by solving~\eqref{eq:quartic_u} and selecting the root on the unit circle that minimizes $\Re\{\alpha_\ell^{(\tau)}u+\beta_\ell^{(\tau)}u^2\}$. Finally, we quantize the continuous minimizer onto the finite-resolution alphabet:
\begin{equation}
\label{eq:coord_update_rule}
\phi_\ell^{(\tau+1)} \leftarrow 
\Pi_{\mathcal A}\left(\phi_{\ell,\rm cont}^{(\tau+)}\right)~ (\ell=1,\cdots,M_o).
\end{equation}
Sweeping $\ell=1,\cdots,M_o$ once constitutes one full CD-update pass. At termination, we output $\bm\phi^\star$ and set $\theta_{i_\ell}^\star=\angle \phi_\ell^\star~(\forall \ell)$. The resulting update is computationally efficient because each coordinate update involves only a scalar quartic root-finding and a single quantization operation, while preserving a non-increasing objective sequence. Furthermore, the convergence of the proposed CD-based update of $\boldsymbol\theta$ is established by the following lemma:
\begin{lemma}
\label{lem:cd_convergence_cont}
Consider the continuous-phase refinement problem
\begin{equation}
\label{eq:cont_cd_prob}
\min_{\bm\phi\in\mathbb C^{M_o}}~\mathcal L(\bm\phi)
~ \text{s.t.}~ |\phi_\ell|=1~(\ell=1,\cdots,M_o),
\end{equation}
where $\mathcal L(\bm\phi)$ is given in~\eqref{eq:L_phi_form}.
At each CD step, the $\ell$th coordinate is updated by solving the one-dimensional subproblem in~\eqref{eq:phi_cont_problem}, i.e.,
\begin{equation}
\label{eq:exact_cd_update}
\phi_\ell^{(t+1)} \in \argmin_{|\phi|=1} \mathcal L\left(\phi_1^{(t+1)},\cdots,\phi_{\ell-1}^{(t+1)},\phi,\phi_{\ell+1}^{(t)},\cdots,\phi_{M_o}^{(t)}\right),
\end{equation}
while keeping all other entries fixed. Then:
\begin{enumerate}
\item \textbf{(Monotonicity)} The objective values are non-increasing:
\begin{equation}
\label{eq:cd_monotone}
\mathcal L(\bm\phi^{(t+1)})\le \mathcal L(\bm\phi^{(t)})~(\forall t\in\mathbb N \cup \{0\}).
\end{equation}
\item \textbf{(Convergence)} $\{\mathcal L(\bm\phi^{(t)})\}$ converges to a finite limit.
\item \textbf{(Coordinate-wise stationarity of limit points)} Every limit point $\bm\phi^\star$ of $\{\bm\phi^{(t)}\}$ satisfies
\begin{equation}
\label{eq:cd_coord_opt}
\phi_\ell^\star \in \argmin_{|\phi|=1} 
\mathcal L\left(\phi_1^\star,\cdots,\phi_{\ell-1}^\star,\phi,\phi_{\ell+1}^\star,\cdots,\phi_{M_o}^\star\right)~(\forall \ell),
\end{equation}
i.e., $\bm\phi^\star$ is a coordinate-wise minimizer on the product of unit circles.
\end{enumerate}
\end{lemma}
\begin{proof}
Fix an outer iteration index $t$ and consider the inner sweep over $\ell=1,\cdots,M_o$. By construction, the update~\eqref{eq:exact_cd_update} yields
\begin{equation}
\label{eq:cd_step_descent}
\begin{aligned}
&\mathcal L\left(\phi_1^{(t+1)},\cdots,\phi_{\ell}^{(t+1)},\phi_{\ell+1}^{(t)},\cdots,\phi_{M_o}^{(t)}\right)\\
&\le\mathcal L\left(\phi_1^{(t+1)},\cdots,\phi_{\ell-1}^{(t+1)},\phi_{\ell}^{(t)},\phi_{\ell+1}^{(t)},\cdots,\phi_{M_o}^{(t)}\right).
\end{aligned}
\end{equation}
Chaining~\eqref{eq:cd_step_descent} over $\ell=1,\cdots,M_o$ gives~\eqref{eq:cd_monotone}.

Next, note that $\mathcal L(\bm\phi)\ge 0$ for all feasible $\bm\phi$ because of~\eqref{eq:L_def}. Hence, $\{\mathcal L(\bm\phi^{(t)})\}$ is a non-increasing sequence bounded below by $0$, and therefore it converges to a finite limit, proving the convergence claim.

Finally, the feasible set $\mathcal S\triangleq\{\bm\phi\in\mathbb C^{M_o}:|\phi_\ell|=1~(\forall \ell)\}$ is compact. Thus, $\{\bm\phi^{(t)}\}\subset\mathcal S$ admits at least one limit point. Let $\bm\phi^\star$ be any limit point, and consider a subsequence $\{\bm\phi^{(t_j)}\}$ such that $\bm\phi^{(t_j)}\to \bm\phi^\star$. Because $\mathcal L(\bm\phi)$ is continuous in $\bm\phi$ and each coordinate update solves an exact one-dimensional minimization over the compact set $\{|\phi|=1\}$, the inequality in~\eqref{eq:cd_step_descent} passes to the limit along the subsequence, yielding~\eqref{eq:cd_coord_opt} for every $\ell$. This proves the third claim.
\end{proof}
We further verify numerically that the proposed CD update yields a monotonically decreasing objective sequence in Section~\ref{rpf}, confirming its descent property and validating the theoretical convergence behavior.

\subsection{Overall AO Framework}
\label{subsec:overall_ao}
We now summarize the overall AO procedure for solving~\eqref{prob:fris_quad_leakage}. Starting from a feasible initialization $(\Gamma^{(0)},\boldsymbol\theta^{(0)},\mathbf w^{(0)})$ and AO iteration index $t$, the proposed method alternates among: (i) a closed-form update of $\mathbf w^{(t)}$ via~\eqref{prob:w_update_real}; (ii) a CEM-based update of $\Gamma^{(t)}$; and (iii) a CD update of $\boldsymbol\theta^{(t)}$. The resulting objective sequence is non-increasing across AO iterations, and the algorithm terminates when the relative objective improvement falls below a prescribed tolerance.

\begin{algorithm}[t]
\caption{Proposed AO Framework}
\label{alg:AO_FRIS_RAR}
\begin{algorithmic}[1]
\Require $\mathbf H_{\rm UR},\mathbf H_{\rm RV},\mathbf H_{\rm UV}$, $P$, $M_o$, $\mathcal Q$, $(K,K_{\rm CEM},\rho,\alpha)$, phase-refinement passes $T_\theta$, tolerance $\varepsilon$.

\State \textbf{Initialization:} Choose any feasible $\Gamma^{(0)}$ with $|\Gamma^{(0)}|=M_o$, $\boldsymbol\theta^{(0)}$ with $\theta_{i_\ell}^{(0)}\in\mathcal Q$, and $\mathbf w^{(0)}$ with $\|\mathbf w^{(0)}\|_2^2=P$.
\State Evaluate $\mathcal L^{(0)}$ via~\eqref{eq:L_final}.

\While{$|\mathcal L^{(t)}-\mathcal L^{(t+1)}|\ge\varepsilon$}

  \Statex \textbf{1) Update $\mathbf w$ for fixed $(\Gamma^{(t)},\boldsymbol\theta^{(t)})$:}
  \State Form $\mathbf G$ by~\eqref{eq:G_explicit} and set $\tilde{\mathbf w}^{(t+1)}\leftarrow \sqrt{P} \mathbf v_{\min}(\mathbf G)$.
  \State Recover $\mathbf w^{(t+1)}$ by $\tilde{\mathbf w}^{(t+1)}$.

  \Statex \textbf{2) Update $\Gamma$ for fixed $(\boldsymbol\theta^{(t)},\mathbf w^{(t+1)})$:}
  \State Initialize uniform $\mathbf p^{(0)}$.
  \For{$k=1,\cdots,K_{\rm CEM}$}
    \State \multiline{Sample $K$ sets $\{\Gamma_i^{(k)}\}_{i=1}^K$ of size $M_o$ without replacement according to $\mathbf p^{(k-1)}$.}
    \State Evaluate $f(\Gamma_i^{(k)})=\mathcal L(\Gamma_i^{(k)},\boldsymbol\theta^{(t)},\mathbf w^{(t+1)})$ by~\eqref{eq:L_final}.
    \State Determine $\mathcal E^{(k)}$ by~\eqref{eq:elite_set}.
    \State Update $\widehat{\mathbf p}^{(k)}$ by~\eqref{eq:pmf_update_raw} and smooth $\mathbf p^{(k)}$ by~\eqref{eq:pmf_smooth}.
  \EndFor
  \State Set $\Gamma^{(t+1)}\leftarrow \mathrm{Top}_{M_o}(\mathbf p^{(K_{\rm CEM})})$ as in~\eqref{eq:Gamma_from_p}.

  \Statex \textbf{3) Update $\boldsymbol\theta$ for fixed $(\Gamma^{(t+1)},\mathbf w^{(t+1)})$:}
  \State Initialize $\bm\phi^{(0)}\leftarrow [e^{j\theta_{i_1}^{(t)}}~\cdots~e^{j\theta_{i_{M_o}}^{(t)}}]^{\mathrm T}$ over $\Gamma^{(t+1)}$.
  \For{$\tau=0,1,\cdots,T_\theta-1$}
    \For{$\ell=1,2,\cdots,M_o$}
      \State \multiline{Form $\mathbf Y_{-\ell}$ by~\eqref{eq:Y_minus_ell} and set $(\alpha_\ell,\beta_\ell)$ by~\eqref{eq:alpha_beta_def}.}
      \State Solve~\eqref{eq:quartic_u} and form $\mathcal U_\ell$ per~\eqref{eq:U_set}.
      \State Choose $u_\ell$ and set $\phi_{\ell,\rm cont}$ by~\eqref{eq:u_star}.
      \State Quantize $\phi_\ell\leftarrow \Pi_{\mathcal A}(\phi_{\ell,\rm cont})$ via~\eqref{eq:phi_quantize}.
    \EndFor
  \EndFor
  \State Set $\boldsymbol\theta^{(t+1)}\leftarrow \angle\bm\phi$.

  \State Compute $\mathcal L^{(t+1)}=\mathcal L(\Gamma^{(t+1)},\boldsymbol\theta^{(t+1)},\mathbf w^{(t+1)})$ via~\eqref{eq:L_final}.
 \State $t\leftarrow t+1$
\EndWhile
\State \Return $(\Gamma^\star,\boldsymbol\theta^\star,\mathbf w^\star)\leftarrow (\Gamma^{(t+1)},\boldsymbol\theta^{(t+1)},\mathbf w^{(t+1)})$.
\end{algorithmic}
\end{algorithm}

\subsection{Signal Detection}
\label{subsec:signal_detection}
After obtaining $(\Gamma^\star,\boldsymbol\theta^\star,\mathbf w^\star)$, the effective channel becomes fixed as $
\mathbf H_{\rm eq}^\star
\triangleq
\mathbf H_{\rm eq}(\Gamma^\star,\boldsymbol\theta^\star)$. The atomic receiver measurement is then expressed as in~\eqref{prpr}:
\begin{equation}
\label{cab}
\mathbf y-|\mathbf r|
\approx
\mathbf g^\star s+\widetilde{\mathbf n},
\end{equation}
where $\mathbf g^\star
=
\mathbf D_r \mathbf H_{\rm eq}^\star \mathbf w^\star
\in\mathbb R^{N_r\times1}$, since the proposed leakage-aware design suppresses the residual quadrature component: $\Im\left\{\mathbf D_r \mathbf H_{\rm eq}^\star \mathbf w^\star s\right\}\approx \mathbf 0$, which makes the reference-aligned signal nearly real-valued. 

Accordingly, symbol detection reduces to a one-dimensional least-squares (LS) problem
\begin{equation}
\label{eq:ls_scalar_detect}
\widehat{s}
=
\arg\min_{s\in\mathcal S}
\left\|
(\mathbf y-|\mathbf r|)-\mathbf g^\star s
\right\|_2^2,
\end{equation}
where $\mathcal S$ denotes the symbol constellation. An equivalent implementation first computes the unconstrained LS estimate
\begin{equation}
\label{eq:s_hat_unconstrained}
\widetilde{s}
=
\frac{\mathbf g^{\star\mathrm{T}}(\mathbf y-|\mathbf r|)}
{\|\mathbf g^\star\|_2^2},
\end{equation}
and then determines the final decision by nearest-neighbor slicing,
\begin{equation}
\label{eq:slicer_detect}
\widehat{s}
=
\arg\min_{s\in\mathcal S}
|\widetilde{s}-s|^2.
\end{equation}


\section{Computational Complexity}
\label{sec:complexity}
The computational complexity of the proposed framework can be decomposed into updating each variables and signal detection:
\subsection{Update of $\mathbf w$}
The update of $\mathbf w$ requires forming $\mathbf R=\mathbf H^{*}\mathbf H\in\mathbb C^{N_t\times N_t}$ and $\mathbf B=\mathbf H^{\mathrm T}\mathbf D\mathbf H\in\mathbb C^{N_t\times N_t}$. Computing $\mathbf R$ and $\mathbf B$ both cost $\mathcal O(N_r N_t^2)$ since $\mathbf D$ is diagonal. Forming the real-augmented matrix $\mathbf G\in\mathbb R^{2N_t\times 2N_t}$ is $\mathcal O(N_t^2)$, while obtaining $\mathbf v_{\min}(\mathbf G)$ via an eigen-decomposition of a $(2N_t)\times(2N_t)$ symmetric matrix costs $\mathcal O(N_t^3)$. Therefore, the beamformer-update complexity per AO iteration is
\begin{equation}
\label{eq:complex_w}
\mathcal O\left(N_r N_t^2 + N_t^3\right).
\end{equation}

\subsection{Update of $\Gamma$}
At each CEM iteration, we sample $K$ candidate sets of size $M_o$ without replacement according to $\mathbf p^{(k-1)}$ and evaluate $f(\Gamma)$ for each sampled set. The dominant cost comes from evaluating~\eqref{eq:L_final}, which requires forming $\mathbf H_{\rm eq}(\Gamma,\boldsymbol\theta)=\mathbf H_{\rm RV}\mathbf S_\Gamma\mathbf \Phi \mathbf S_\Gamma^{\mathrm T}\mathbf H_{\rm UR}+\mathbf H_{\rm UV}$. Using $\mathbf H_{\rm RV}\mathbf S_\Gamma\in\mathbb C^{N_r\times M_o}$ and $\mathbf S_\Gamma^{\mathrm T}\mathbf H_{\rm UR}\in\mathbb C^{M_o\times N_t}$, $\mathbf H_{\rm eq}(\Gamma,\boldsymbol\theta)$ can be formed as $(\mathbf H_{\rm RV}\mathbf S_\Gamma)\mathbf \Phi (\mathbf S_\Gamma^{\mathrm T}\mathbf H_{\rm UR})$ with complexity $\mathcal O(N_r M_o N_t)$. Given $\mathbf H_{\rm eq}$, evaluating~\eqref{eq:L_final} further requires computing $\mathbf H_{\rm eq}\mathbf Q\mathbf H_{\rm eq}^{*}$ and $\mathbf H_{\rm eq}\mathbf P\mathbf H_{\rm eq}^{\mathrm T}$ traces. Since $\mathbf Q=\mathbf w\mathbf w^{*}$ and $\mathbf P=\kappa\mathbf w\mathbf w^{\mathrm T}$ are rank-one, these traces can be computed via vector products (e.g., $\|\mathbf H_{\rm eq}\mathbf w\|_2^2$) in $\mathcal O(N_r N_t)$. Hence, one objective evaluation costs $\mathcal O\left(N_r M_o N_t + N_r N_t\right)=\mathcal O\left(N_r M_o N_t\right)$. Sorting $K$ objective values costs $\mathcal O(K\log K)$, and updating the PMF using~\eqref{eq:pmf_update_raw} costs $\mathcal O(K M_o)$ by counting port occurrences among elite samples. Therefore, the overall complexity of the CEM update per AO iteration is
\begin{equation}
\label{eq:complex_cem}
\mathcal O\left( K_{\rm CEM}\Big(KN_r M_o N_t + K\log K + K M_o\Big)\right).
\end{equation}

\subsection{Update of $\boldsymbol\theta$}
For fixed $(\Gamma,\mathbf w)$, the CD procedure performs sequential updates over $\ell=1,\cdots,M_o$. Each coordinate update requires evaluating $(\alpha_\ell,\beta_\ell)$ in~\eqref{eq:alpha_beta_def} and solving~\eqref{eq:quartic_u}, whose root-finding cost is constant. The main computational burden arises from evaluating the trace expressions defining $(\alpha_\ell,\beta_\ell)$, which involve $\mathbf G_\ell\in\mathbb C^{N_r\times N_t}$ and $\mathbf Y_{-\ell}$. Again by exploiting the rank-one structures $\mathbf Q$, $\mathbf P$, and $\mathbf G_\ell=\mathbf a_\ell\mathbf b_\ell^{\mathrm T}$ with $\mathbf a_\ell=\mathbf H_{\rm RV}\mathbf e_{i_\ell}$ and $\mathbf b_\ell^{\mathrm T}=\mathbf e_{i_\ell}^{\mathrm T}\mathbf H_{\rm UR}$, the trace terms can be expressed in terms of vector inner products. Furthermore, maintaining the auxiliary vector $\mathbf y_v\triangleq \mathbf Y(\bm\phi)\mathbf w$ enables efficient reuse across coordinates via the relation $\mathbf Y_{-\ell}\mathbf w=\mathbf y_v-\phi_\ell(\mathbf G_\ell\mathbf w)$, thereby reducing each coordinate update to $\mathcal O(N_r+N_t)$ arithmetic after an $\mathcal O(N_rN_t)$ preprocessing step per CD pass to compute $\mathbf y_v$. Therefore, the overall computational complexity of the phase refinement stage with $T_\theta$ CD passes is given by
\begin{equation}
\label{eq:complexity_theta}
\mathcal O\left(T_\theta\big(N_rN_t+M_o(N_r+N_t)\big)\right).
\end{equation}

\subsection{Signal Detection}
\label{subsubsec:detection_complexity}
Based on the linearized detection model in~\eqref{eq:ls_scalar_detect}-\eqref{eq:s_hat_unconstrained}, symbol detection only requires computing the scalar least-squares estimate in~\eqref{eq:s_hat_unconstrained}. First, computing $\mathbf z=\mathbf y-|\mathbf r|$ requires $\mathcal O(N_r)$ operations. The numerator $\mathbf g^{\star\mathrm T}\mathbf z$ is a vector inner product, which also requires $\mathcal O(N_r)$ operations, while the denominator $\|\mathbf g^\star\|_2^2$ can be precomputed once after the optimization stage with complexity $\mathcal O(N_r)$. The final nearest-neighbor slicing step in~\eqref{eq:slicer_detect} requires $\mathcal O(|\mathcal S|)$ operations. Therefore, the overall detection complexity per symbol scales as $\mathcal O(N_r+|\mathcal S|)$.

\subsection{Overall Complexity}
Combining~\eqref{eq:complex_w}-\eqref{eq:complexity_theta}, the total complexity of the proposed AO algorithm with $T_{\rm AO}$ AO iterations scales as
\begin{equation}
\label{eq:complex_total}
\begin{aligned}
\mathcal O\Big(T_{\rm AO}\Big[&N_r N_t^2 + N_t^3+K_{\rm CEM}\big(KN_r M_o N_t + K\log K\big)\\
&+T_\theta (N_rN_t+M_o(N_r+N_t))+|\mathcal S|\Big]\Big).
\end{aligned}
\end{equation}

\section{Simulation Results}
\subsection{Parameter Setup}
We first define the received SNR as
\begin{equation}
\label{resnr}
\mathrm{SNR}
\triangleq
\frac{\mathbb{E}\left[\big\|\mathbf H_{\mathrm{eq}}\mathbf x\big\|_2^2\right]}
{\mathbb{E}\left[\|\mathbf n\|_2^2\right]}.
\end{equation}
To characterize the relative strength of the reference signal, we further introduce the reference-to-signal ratio (RSR)~\cite{Precoding_atomicMIMO}, defined as
\begin{equation}
\label{rsrdef}
\mathrm{RSR}
\triangleq
\frac{\mathbb{E}\left[\|\mathbf r\|_2^2\right]}
{\mathbb{E}\left[\big\|\mathbf H_{\mathrm{eq}}\mathbf x+\mathbf n\big\|_2^2\right]},
\end{equation}
which measures the reference signal power relative to the combined received communication signal and noise.
\begin{figure}[t]
  \begin{center}
    \includegraphics[width=0.8\columnwidth,keepaspectratio]{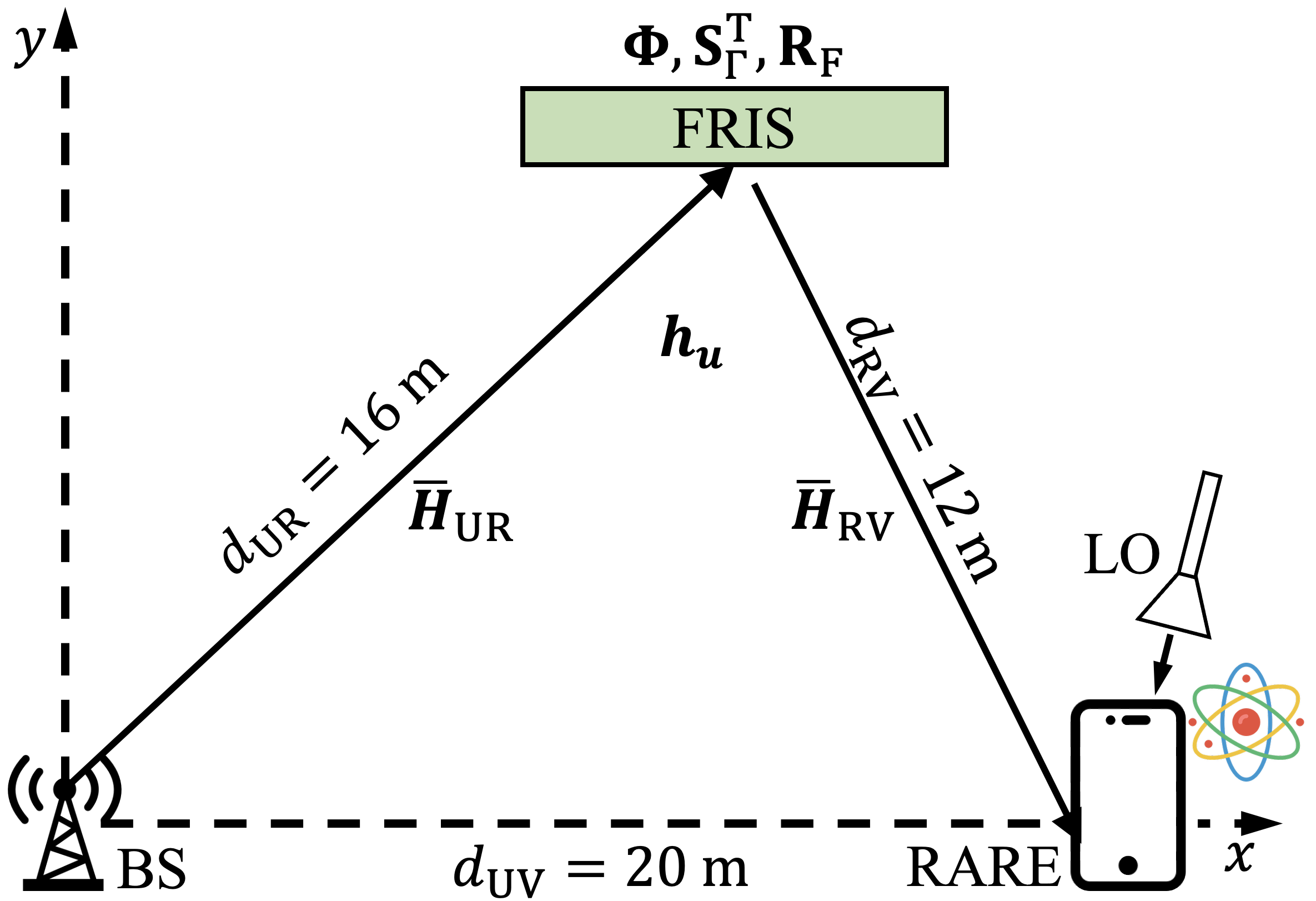}
    \caption{Simulation setup of the proposed system.}
    \label{fig_simset}
  \end{center}
\end{figure}

Unless otherwise specified, all simulations employ the parameter settings listed in Table~\ref{tabsim}. The distances of the BS-FRIS, FRIS-RARE, and BS-RARE links are set to $d_{\mathrm{UR}}=16$~m, $d_{\mathrm{RV}}=12$~m, and $d_{\mathrm{UV}}=20$~m, respectively, which is depicted in Fig.~\ref{fig_simset}. For the RARE configuration, the Rydberg states $52D_{5/2}$ and $53P_{3/2}$ are used to detect signals at $f=5$~GHz. Based on~\cite{rydpar}, the RF transition dipole moment between these states is given by $\boldsymbol{\mu}_{\rm RF}=[0,1785.916qa_0, 0]^{\mathrm T}$, where $a_0=5.292\times10^{-11}$~m denotes the Bohr radius and $q=1.602\times10^{-19}$~C represents the elementary charge. The polarization vectors ${\boldsymbol{\epsilon}_{m,u,\ell}}$, ${\boldsymbol{\epsilon}_{m,n,\ell}}$, and ${\boldsymbol{\epsilon}_{m,t,\ell}}$ are randomly generated on unit circles orthogonal to their respective propagation directions. Each simulation result is averaged over $10^3$ independent Monte Carlo realizations. For the LO-RARE link, due to the sufficiently small separation between the LO and RAR, we approximate $\{\rho_{b,m}\}$ by a common amplitude $\rho_b$, while modeling $\{\phi_{b,m}\}$ as uniformly distributed over $[0,2\pi)$.

\begin{table}[t]
\centering
\caption{Simulation Parameters}
\label{tabsim}
\begin{tabular}{l c}
\toprule
\textbf{Parameter} & \textbf{Value} \\
\midrule
Number of RARE elements $N_r$ & 36 \\
Number of BS antennas $N_t$ & 16 \\
Received SNR in~\eqref{resnr} & 7~dB \\
Number of FRIS elements $(N, M_o)$ & (36, 9) \\
Normalized FRIS aperture $W_x$ & 2 \\
RSR in~\eqref{rsrdef} & 10~dB \\
Rician $K$-factor $\kappa$ & 2 \\
Modulation scheme (unless referred) & 4-QAM\\
\bottomrule
\end{tabular}
\end{table}

To evaluate the effectiveness of the proposed scheme, we consider the following benchmark methods.
\begin{itemize}
\item \textbf{ZF (Known Phase)}: An ideal reference case where the receiver has prior knowledge of the phase of $\mathbf y$
and directly performs zero-forcing (ZF) detection for demodulation (equivalent to zero-quadrature condition in~\eqref{prpr}).
\item \textbf{LS-Based Exhaustive Search}: An exhaustive search-based detection algorithm is applied after applying the proposed AO framework based on~\cite{atomicjsac}:
\begin{equation}
\label{Jsacls}
\min_{s} \left\|\mathbf y-|\mathbf H_{\rm eq}\mathbf x+\mathbf r|
\right\|_2^2.
\end{equation}
\item \textbf{RIS-Assisted Atomic-MIMO Benchmark}: Algorithm in~\cite{RIS_atomic_MIMO} with $M_o$-element conventional RIS.
\end{itemize}


\begin{figure}[t]
    \centering
    \subfloat[]{%
        \includegraphics[width=0.4\textwidth]{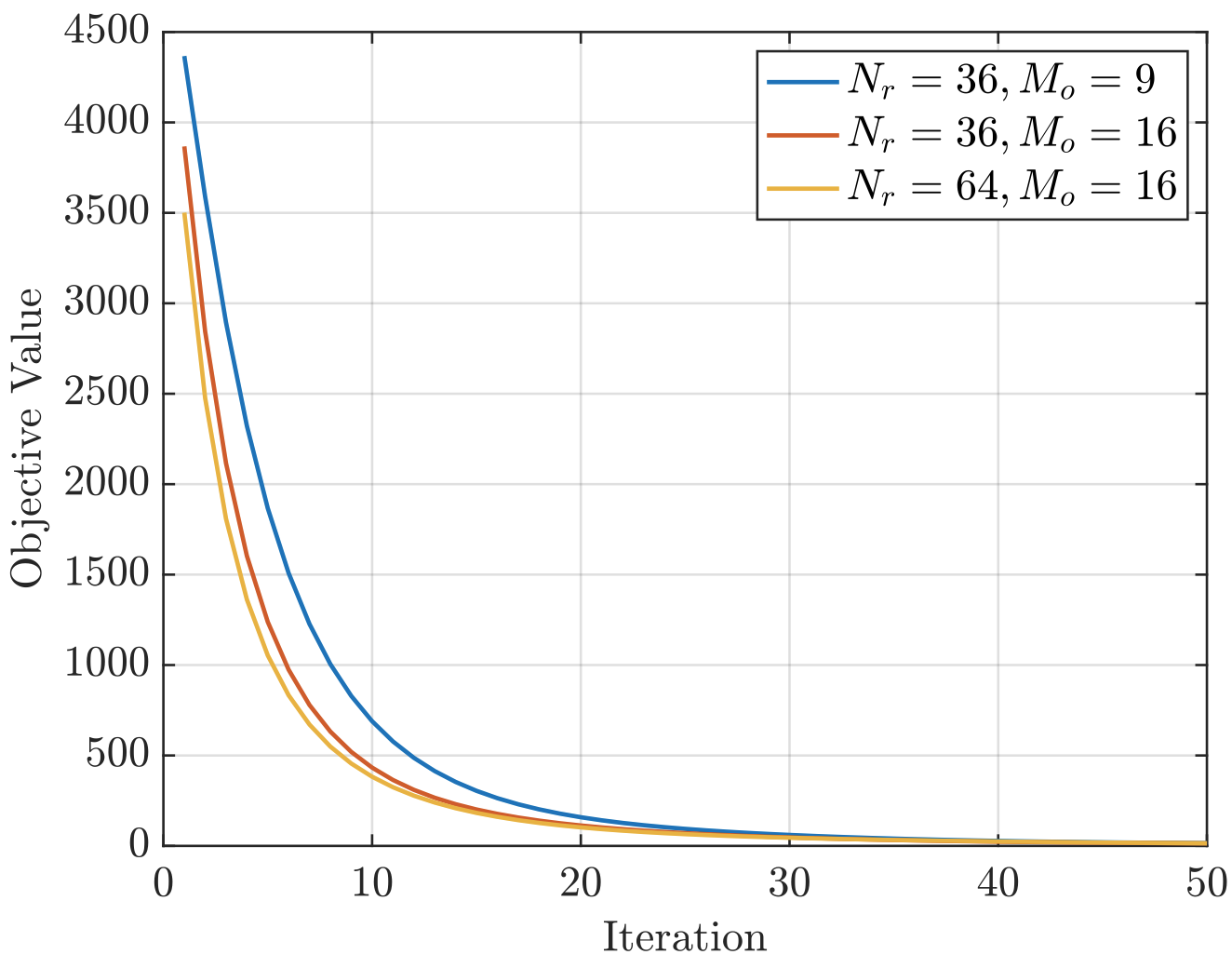}
    \label{fig_conv}
    }
    \vfill
    \subfloat[]{%
        \includegraphics[width=0.4\textwidth]{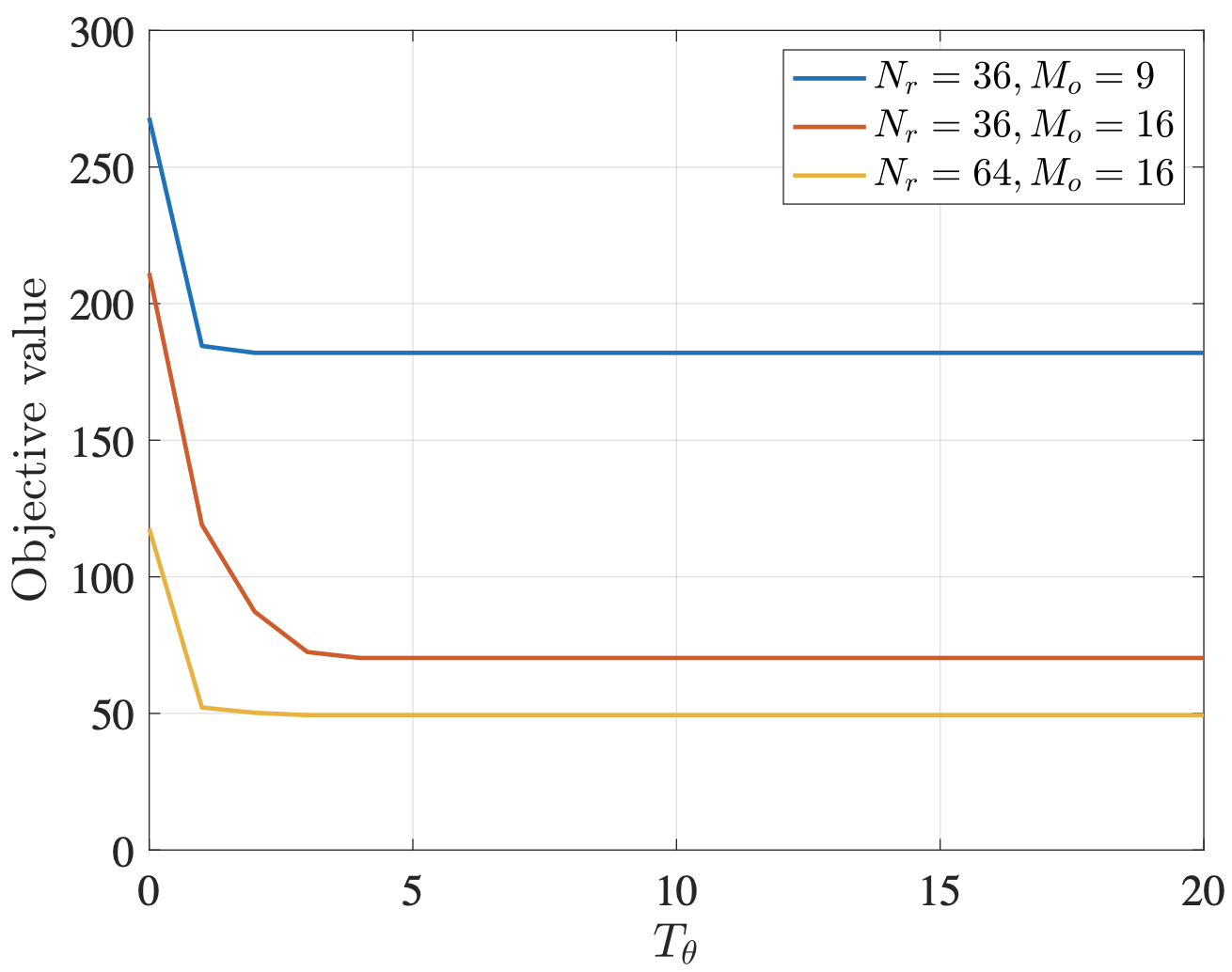}
        \label{fig_cdconv}%
    }
    \caption{Convergence of objective function under (a) overall AO framework (b) inner framework based on CD.}
    \label{fig_0}
\end{figure}
\subsection{Reliability of Proposed Framework}
\label{rpf}
To validate the convergence of the proposed AO framework, we plot $\mathcal L$ in Fig.~\ref{fig_conv} with respect to outer iteration under different $N_r$ and $M_o$. As observed, the proposed algorithm converges rapidly and stably within approximately 20-30 iterations in all cases, demonstrating its effectiveness and robustness. In addition, the convergence behavior of the inner CD phase-refinement procedure is illustrated in Fig.~\ref{fig_cdconv}.\footnote{Since the update of $\mathbf w$ admits a closed-form eigenvector solution and the CEM-based port selection is known to converge effectively by progressively concentrating the sampling PMF on high-quality configurations~\cite{CEM}, it is sufficient to examine only the convergence behavior of the CD-based phase refinement.} Herein, the objective value decreases monotonically with respect to $T_\theta$, and converges within only a few CD passes. This confirms the effectiveness of the proposed closed-form coordinate update derived from the quartic optimality condition, as well as the stability of the discrete phase quantization step. Moreover, the convergence speed remains fast across different system dimensions in both figures, demonstrating that the proposed framework efficiently suppresses quadrature leakage.

\begin{figure}[t]
  \begin{center}
    \includegraphics[width=0.8\columnwidth,keepaspectratio]{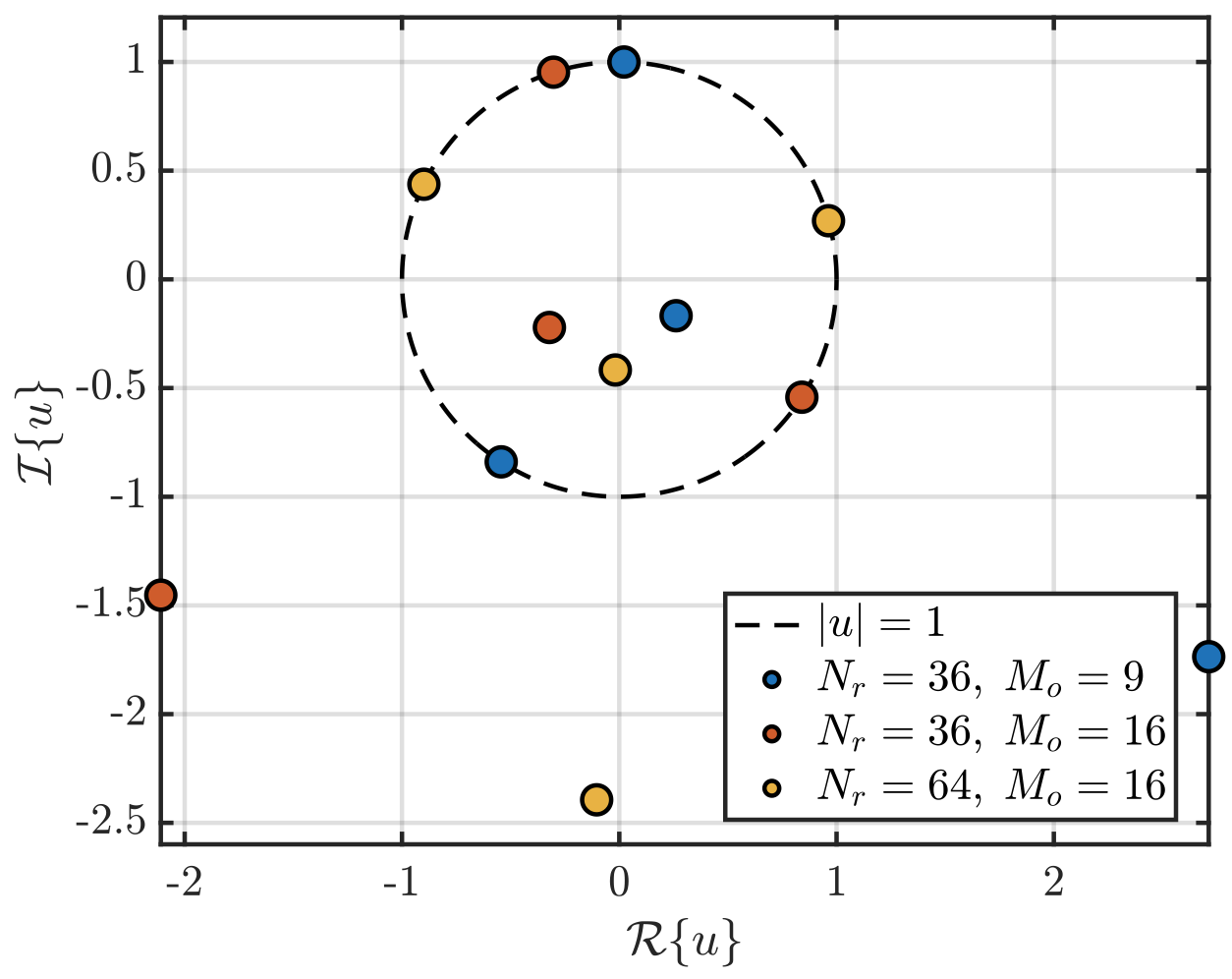}
    \caption{Quartic roots of~\eqref{eq:quartic_u} on the complex plane. The dashed circle denotes $|u|=1$, while the markers indicate the solutions of the quartic optimality condition.}
    \label{fig_quar}
  \end{center}
\end{figure}

Fig.~\ref{fig_quar} illustrates the roots of~\eqref{eq:quartic_u} on the complex plane for a representative channel realization under different $N_r$ and $M_o$. As observed, the quartic equation yields four complex roots in general, but only those located on the unit circle satisfy the feasibility condition. These unit-modulus roots correspond to the stationary points of~\eqref{eq:phi_cont_problem}, while the remaining roots outside the unit circle are infeasible and thus discarded. This result provides direct empirical evidence of the existence of at least one unit-modulus root, (i.e., $\mathcal U_\ell\neq\emptyset$), which is consistent with the theoretical existence guarantee established in Remark~\ref{r11}.

\begin{figure}[t]
    \centering
    \subfloat[4-QAM, $N_r=36$]{%
        \includegraphics[width=0.4\textwidth]{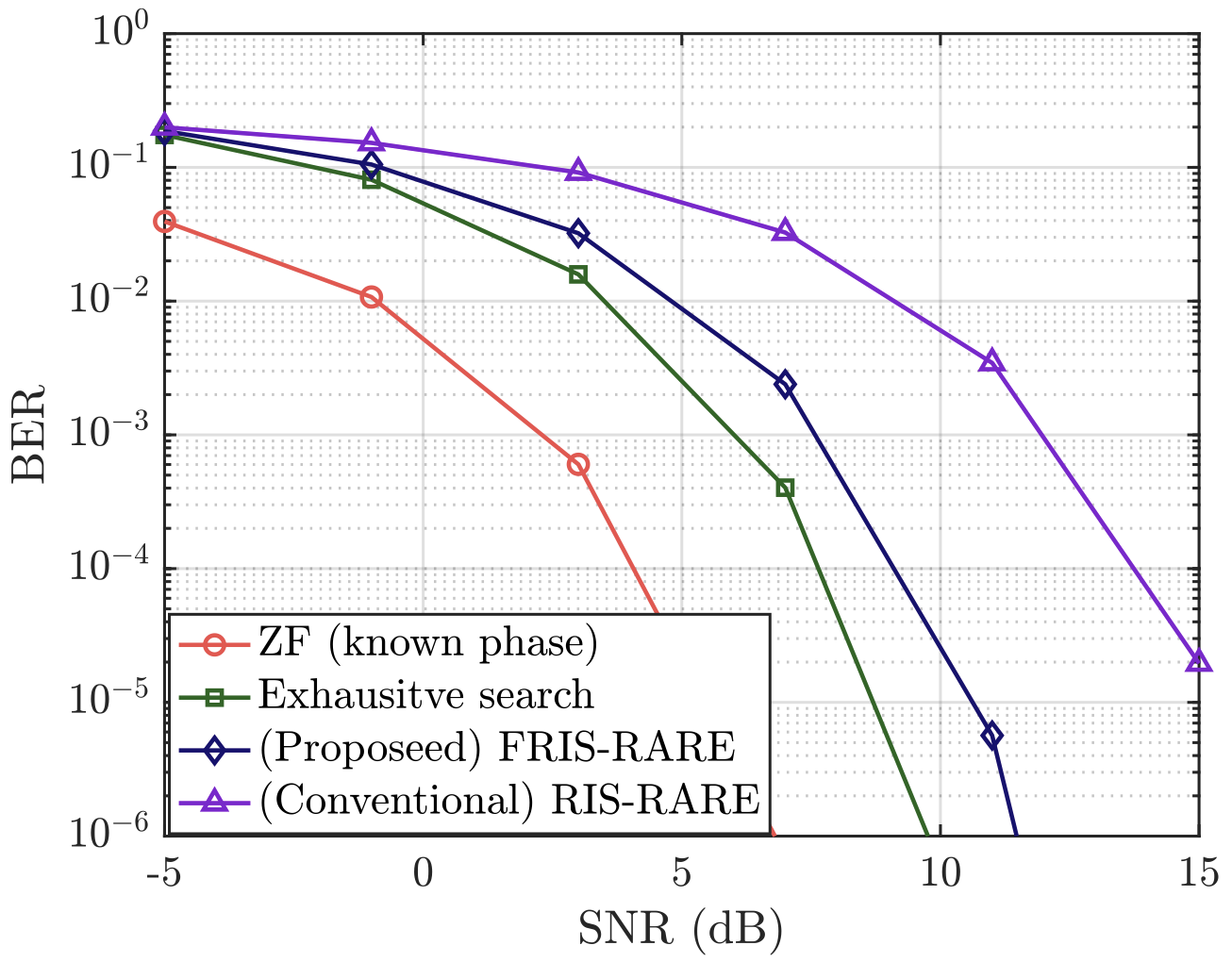}
        \label{fig_436}%
    }
    \vfill
    \subfloat[4-QAM, $N_r=64$]{%
        \includegraphics[width=0.4\textwidth]{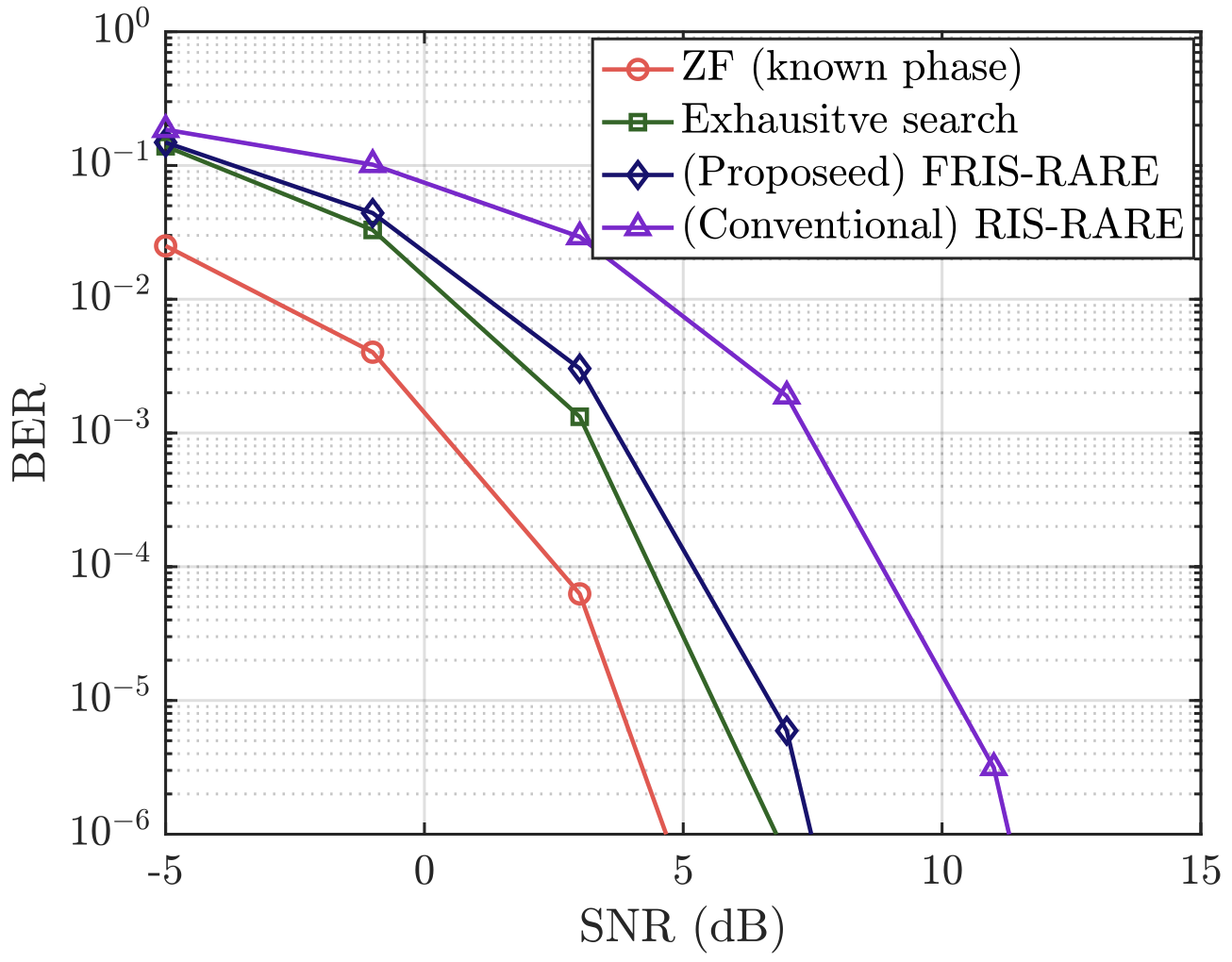}
        \label{fig_464}%
    }
    \vfill
     \subfloat[16-QAM, $N_r=64$]{%
        \includegraphics[width=0.4\textwidth]{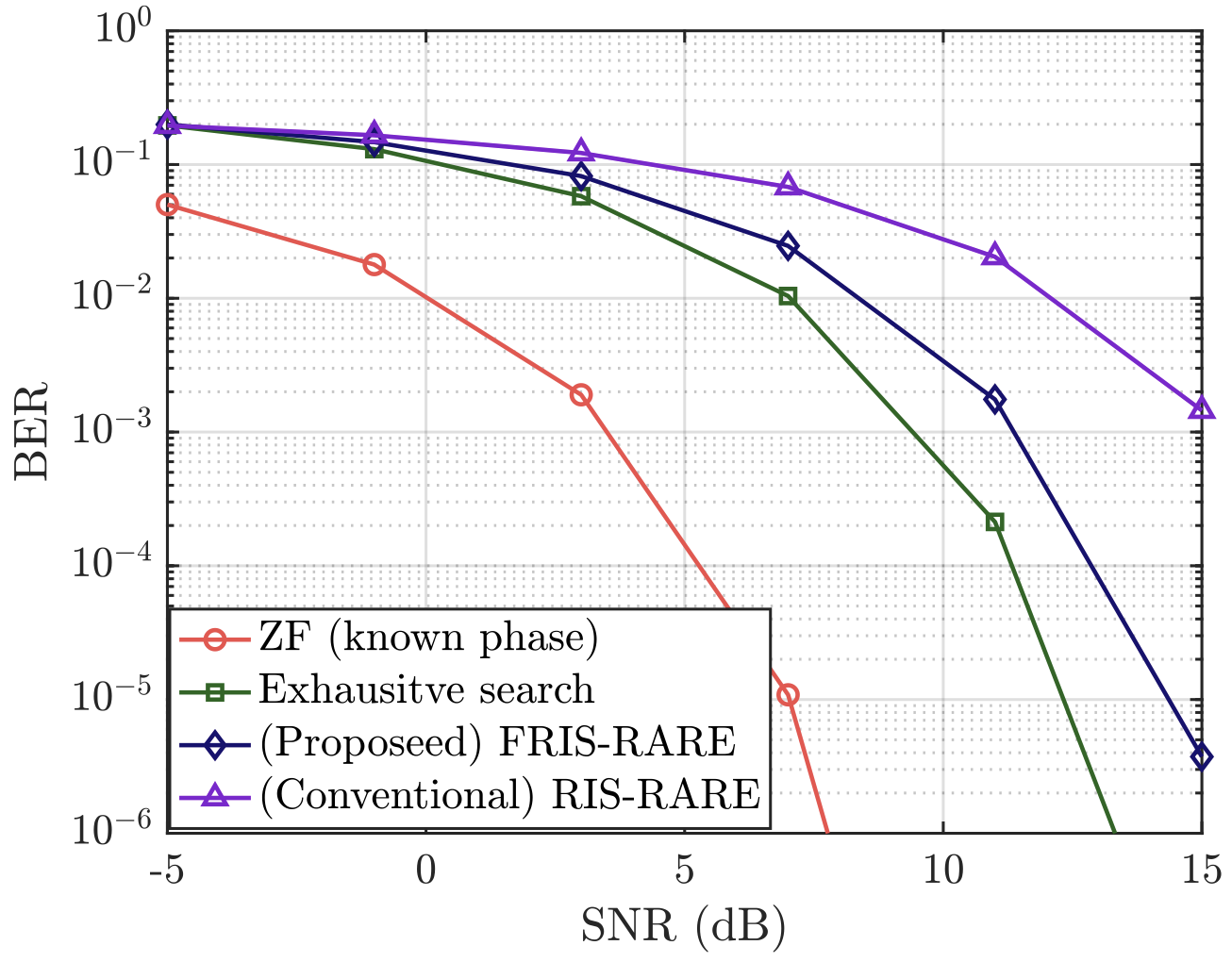}
        \label{fig_1664}%
    }
    \caption{The effect of $\mathrm{SNR}$ on the BER performance under (a) 4-QAM, $N_r=36$, (b) 4-QAM, $N_r=64$, (c) 16-QAM, $N_r=64$.}
    \label{fig_1}
\end{figure}

\subsection{Bit-Error-Rate (BER) Comparisons under Several Effects}
\label{beco}
As illustrated in Fig.~\ref{fig_1}, we evaluate the bit-error-rate (BER) performance as a function of $\mathrm{SNR}$ for different modulation orders and $N_r$. It can be observed that the proposed AO framework with FRIS-RARE consistently outperforms the conventional RIS-RARE scheme across all considered configurations, highlighting its enhanced capability to suppress quadrature leakage and improve detection reliability by exploiting additional spatial DoF; in particular, it achieves approximately 3-5~dB SNR gain over the conventional RIS-RARE scheme at the same target BER. Furthermore, the proposed AO framework achieves performance close to that of the exhaustive search method while incurring substantially lower computational complexity. Although a performance gap remains compared to the ideal ZF detector with known phase, the proposed framework exhibits a similar BER trend despite operating without explicit phase information. This result demonstrates the effectiveness of the proposed FRIS-assisted environment in enabling reliable atomic-MIMO detection.

Herein, comparing Fig.~\ref{fig_436} and Fig.~\ref{fig_464}, it can be observed that increasing $N_r$ from 36 to 64 significantly improves the BER performance across all schemes; for the proposed framework, in particular, it achieves approximately 4~dB SNR gain at the same target BER. This improvement arises because additional RARE elements provide more independent spatial observations, which enhance the effective combining gain and increase the received signal energy relative to noise and residual leakage. As a result, the estimation error variance is reduced, leading to more reliable symbol detection and lower BER. Furthermore, comparing Fig.~\ref{fig_464} and Fig.~\ref{fig_1664}, it is evident that increasing the modulation order from 4 to 16 degrades the BER performance under same SNR, which confirms the expected trade-off between spectral efficiency and detection reliability in atomic-MIMO systems.

\begin{figure}[t]
  \begin{center}
    \includegraphics[width=0.8\columnwidth,keepaspectratio]{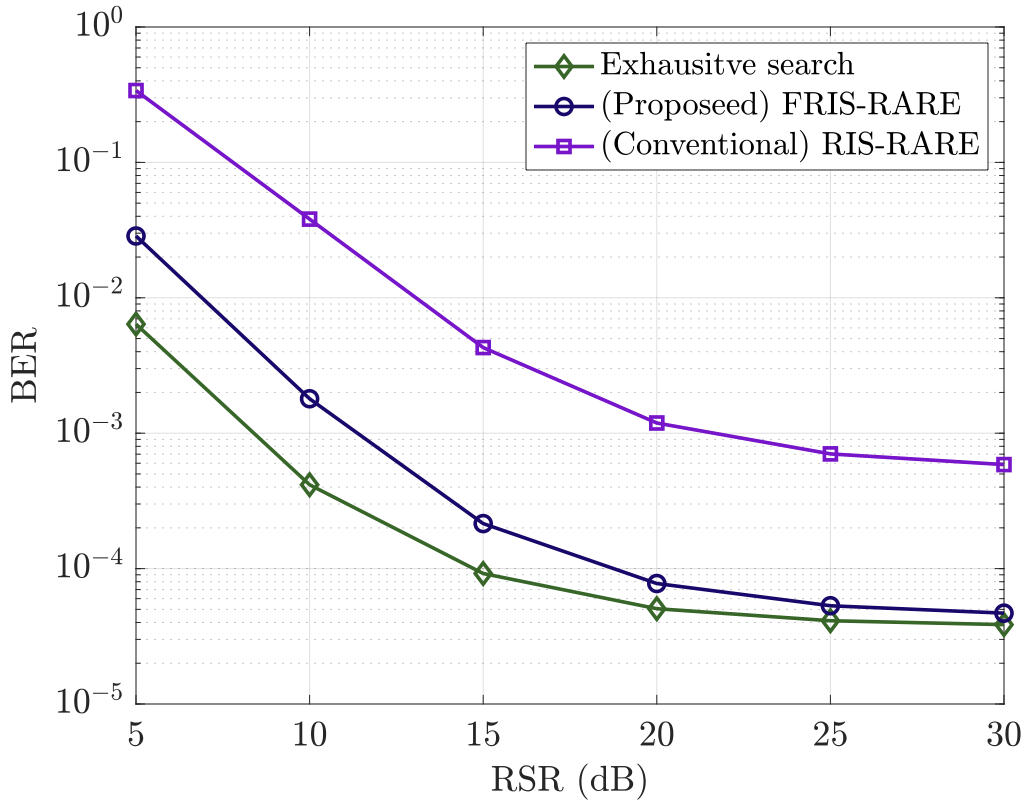}
    \caption{The effect of $\mathrm{RSR}$ on the BER performance.}
    \label{fig_rsr}
  \end{center}
\end{figure}

The impact of $\mathrm{RSR}$ on the BER performance is illustrated in Fig.~\ref{fig_rsr}. Following a similar trend as in Fig.~\ref{fig_1}, the proposed AO framework with FRIS-RARE consistently outperforms the conventional RIS-RARE scheme. In particular, both the proposed framework and the exhaustive search method exhibit substantial BER reduction as the reference signal becomes stronger, where the difference becomes smaller as $\mathrm{RSR}$ increases. This behavior confirms that a sufficiently strong reference signal is essential for reliable detection in RARE, as it effectively mitigates phase ambiguity inherent in magnitude-only observations. Overall, these results highlight the importance of reference signal strength in enabling robust atomic-MIMO detection, while further demonstrating the superiority of the proposed framework.
\begin{figure}[t]
  \begin{center}
    \includegraphics[width=0.8\columnwidth,keepaspectratio]{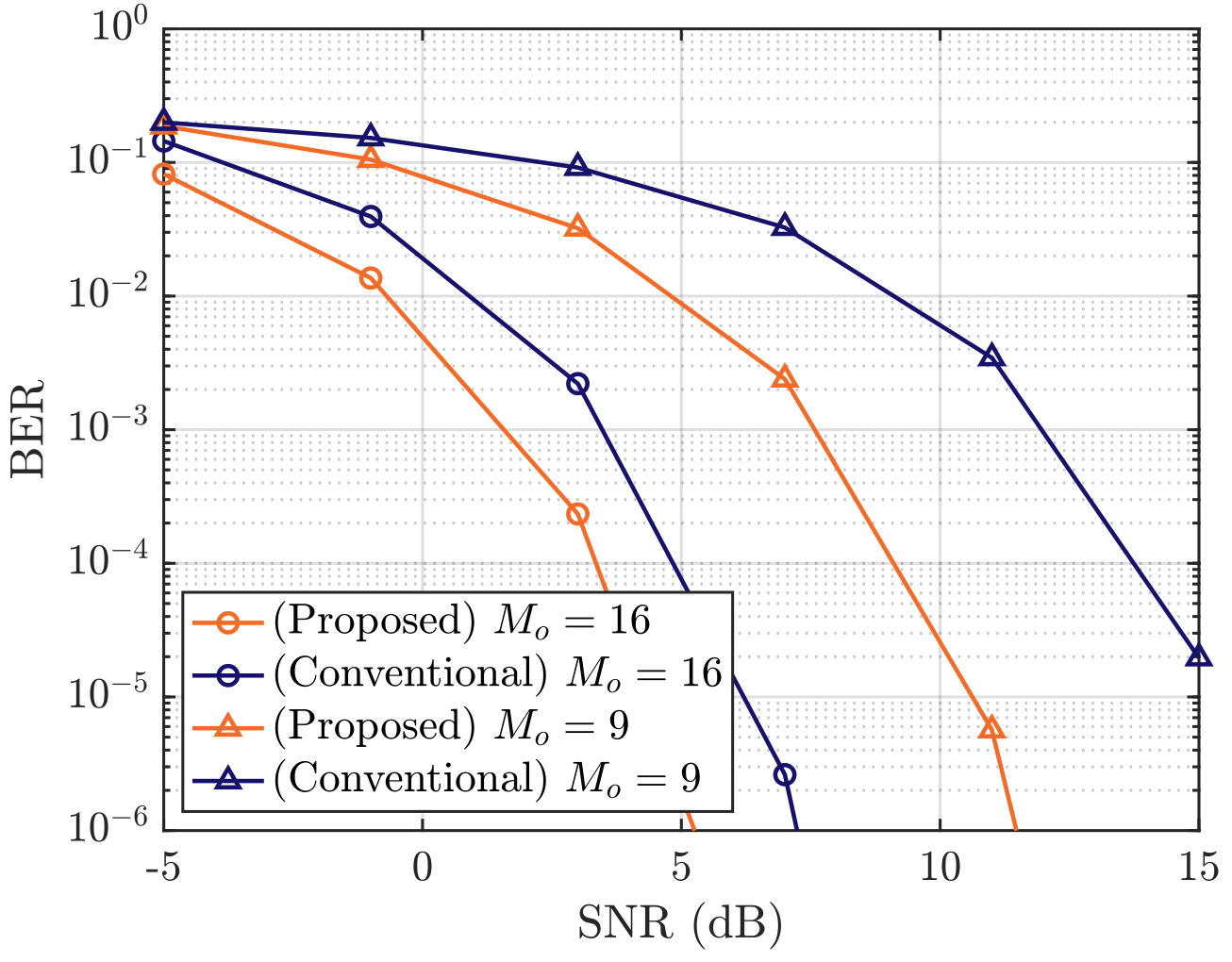}
    \caption{The effect of $\mathrm{SNR}$ on the BER performance under different $M_o$.}
    \label{fig_mo}
  \end{center}
\end{figure}

The effect of $M_o$ on the BER performance is illustrated in Fig.~\ref{fig_mo}. It can be observed that increasing $M_o$ significantly improves the BER performance for both the proposed AO framework with FRIS-RARE and the conventional RIS-RARE scheme; in particular, for the proposed framework, it achieves approximately 4-5~dB SNR gain at the same target BER. This improvement stems from the increased spatial DoF provided by the FRIS, which enable more effective shaping of the equivalent channel and enhanced suppression of quadrature leakage. In particular, the proposed framework consistently outperforms the conventional RIS-RARE for the same $M_o$, highlighting the advantage of the proposed AO-based joint optimization framework in fully exploiting the available FRIS DoF. Moreover, the performance gap between different $M_o$ configurations demonstrates that selecting more FRIS ports provides greater flexibility in controlling the propagation environment, thereby significantly improving detection reliability. These results confirm that the proposed framework effectively leverages the additional FRIS spatial DoF to enhance atomic-MIMO detection performance.

\begin{figure}[t]
  \begin{center}
    \includegraphics[width=0.8\columnwidth,keepaspectratio]{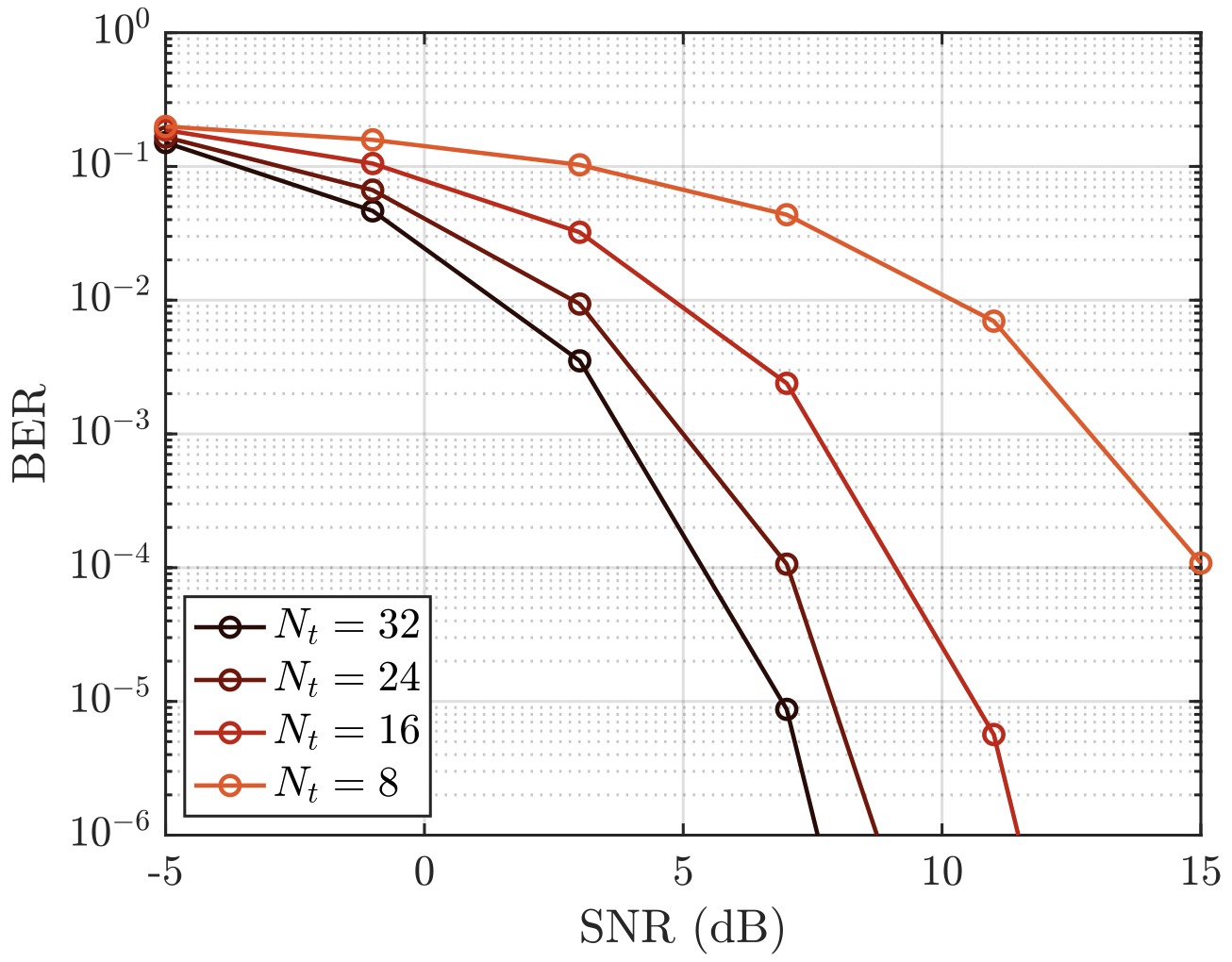}
    \caption{The effect of $\mathrm{SNR}$ on the BER performance under different $N_t$.}
    \label{fig_nt}
  \end{center}
\end{figure}

The effect of $N_t$ on the BER performance is illustrated in Fig.~\ref{fig_nt}. It can be observed that increasing $N_t$ significantly improves the BER performance. In particular, larger transmit antenna arrays achieve the same target BER at substantially lower SNR. For instance, increasing $N_t$ from 16 to 32 provides a noticeable SNR improvement in the moderate-to-low BER region, yielding approximately 4-5~dB gain. This improvement originates from the increased spatial diversity and beamforming gain provided by a larger $N_t$, which improves the reliability of atomic-MIMO detection performance.

\section{Conclusion}
In this paper, we showed that FRIS-assisted atomic-MIMO design requires a fundamentally different perspective from conventional coherent systems. Under magnitude-only heterodyne readout, the optimal propagation environment is not determined by channel strength alone, but depends on its alignment with the receiver's nonlinear measurement structure. By identifying residual quadrature leakage under the strong-reference regime as the dominant structural distortion limiting reliable detection, we established a receiver-induced channel shaping framework. Based on this principle, we formulated a quadrature-leakage-aware joint optimization over transmit beamforming, FRIS port selection, and discrete phase configuration. To solve this challenging mixed design, we developed an efficient AO algorithm that combines a closed-form widely-linear beamformer update, CEM-based combinatorial search, and CD-based phase refinement. The proposed method achieves fast convergence, consistent BER gains, and near-exhaustive performance with significantly reduced complexity. Overall, these results reveal that FRIS is not merely a channel enhancement tool, but a key enabler of receiver-aware propagation control, providing a new design paradigm for robust and practical atomic-MIMO communications in 6G systems.

\bibliographystyle{IEEEtran}
\bibliography{IEEEexample}

\end{document}